\documentclass[amsfonts, amssymb, amsmath,longbibliography,nofootinbib,reqno]{amsart}
\usepackage{graphicx} 
\usepackage{amsaddr}

\usepackage[dvipsnames]{xcolor}

\usepackage{float}
\usepackage[shortlabels]{enumitem}
\usepackage[numbers,sort&compress]{natbib}
\usepackage{braket}
\usepackage{amsthm}
\usepackage{mathtools}
\usepackage{url}
\usepackage{physics}
\usepackage{graphicx}
\usepackage[left=22mm,right=22mm,top=35mm]{geometry} 
\usepackage[T1]{fontenc}
\usepackage{bm}
\usepackage{tabularx}
\newcommand*{\eh}{\mathrm{End\ }\mathcal{H}}
\newcommand*{\Ad}{\mathrm{Ad}}

\def\ad{^{\dagger}}

\def\a{\alpha}

\newcommand{\fsnull}[1]{}
\newcommand{\old}[1]{}

\usepackage[makeroom]{cancel}
\usepackage[toc,page]{appendix}
\definecolor{C1}{RGB}{52, 89, 149}
\definecolor{C2}{RGB}{251, 77, 61}
\definecolor{C3}{RGB}{3, 206, 164}
\definecolor{C4}{RGB}{202, 21, 81}
\definecolor{C5}{RGB}{180,77,155}
\usepackage{hyperref}
\hypersetup{colorlinks=true, linkcolor=C5, citecolor=C5, urlcolor=C5}

\usepackage{tikz}
\tikzset{every picture/.style=remember picture}

\usepackage{pgfplots}

\usepackage[utf8]{inputenc}
\usepackage{graphicx}
\usepackage{xcolor}
\usepackage{amsmath}
\usepackage{amsthm}
\usepackage{bm}
\usepackage{bbm}
\usepackage{comment}
\usepackage{appendix}
\usepackage{mathdots}
\usepackage{lipsum}
\usepackage{verbatim}
\usepackage{nccmath}
\usepackage{amsfonts}
\usepackage{thm-restate}
\usepackage{thmtools}
\usepackage{capt-of}
\usepackage{booktabs}

\usepackage{amssymb}
\usepackage{dsfont}

\newcommand{\img}[2][2.2ex]{%
  \mathrel{\vcenter{\hbox{\includegraphics[height=#1]{#2}}}}%
}

\newcommand{\Var}{{\rm Var}}

\renewcommand{\geq}{\geqslant}
\renewcommand{\leq}{\leqslant}

\def\endh{\mathrm{End\ }\mch}

\newcommand{\ot}{\otimes}
\newcommand{\ts}{^{\otimes 2}}

\newcommand{\bs}{\textsf{BS}}

\newcommand{\sg}{\sigma }

\newcommand{\vph}{\varphi }

\newcommand{\Om}{\Omega }

\newcommand{\F}{\,{}_2F_1}

\newcommand{\Ga}{{\rm Gamma}}
\newcommand{\Ba}{{\rm Beta}}

\DeclareMathOperator*{\expect}{\mathbb{E}}

\def\C{\mathbb{C}}
\newcommand{\mcn}{\mathcal{N}}

\newcommand{\mco}{\mathcal{O}}
\newcommand{\mcb}{\mathcal{B}}
\newcommand{\mcc}{\mathcal{C}}

\newcommand{\mch}{\mathcal{H}}

\newcommand{\mcx}{\mathcal{X}}

\newcommand{\mca}{\mathcal{A}}
\newcommand{\mcd}{\mathcal{D}}
\newcommand{\mce}{\mathcal{E}}

\newcommand{\mct}{\mathcal{T}}
\newcommand{\mbsp}{\mathbb{SP}}
\newcommand{\mbso}{\mathbb{SO}}
\newcommand{\mbsu}{\mathbb{SU}}
\newcommand{\mbo}{\mathbb{O}}
\newcommand{\mbu}{\mathbb{U}}

\newcommand{\mbc}{\mathbb{C}}
\newcommand{\mbr}{\mathbb{R}}
\newcommand{\mbh}{\mathbb{H}}
\newcommand{\mbs}{\mathbb{S}}

\newcommand{\mbe}{\mathbb{E}}
\newcommand{\mst}{\mathsf{T}}

\newcommand{\mff}{\mathfrak{f}}

\newcommand{\mfu}{\mathfrak{u}}

\newcommand{\x}{\boldsymbol{x}}

\def\be{\begin{equation}}
\def\ee{\end{equation}}
\def\bs{\begin{split}}
\def\e{\end{split}}
\def\ba{\begin{eqnarray}}
\def\bea{\begin{eqnarray}}

\def\tea{\end{eqnarray}}
\def\ea{\end{eqnarray}}
\def\eea{\end{eqnarray}}

\def\a{\alpha}
\def\b{\beta}

\def\a{\alpha}

\def\b{\beta}

\def\a{\alpha}
\def\b{\beta}

\newcommand{\id}{\mathds{1}}

\renewcommand{\a}{\alpha}
\renewcommand{\b}{\beta}

\newcommand{\sbraket}[2]{ \langle#1 | #2 \rangle}

\def\sg{\sigma}

\def\be{\begin{equation}}
\def\te{\end{equation}}
\def\ee{\end{equation}}
\def\ba{\begin{eqnarray}}
\def\bea{\begin{eqnarray}}

\def\tea{\end{eqnarray}}
\def\ea{\end{eqnarray}}
\def\eea{\end{eqnarray}}

\newtheorem{theorem}{Theorem}

\begin{document}

\bibliographystyle{IEEEtran}

\title[On the average-case complexity of learning states from the circular and Gaussian ensembles]{On the average-case complexity of learning states\\ from the circular and Gaussian ensembles}
\author{\vspace{-5mm}M\MakeLowercase{axwell} W\MakeLowercase{est}\,\textsuperscript{1,2}}
\address{\textsuperscript{1}Theoretical Division, Los Alamos National Laboratory, Los Alamos, New Mexico 87545, USA\\\phantom{.}\\
\vspace{-2.5mm}
\textsuperscript{2}School of Physics, University of Melbourne, Parkville, Victoria 3010, Australia}

\begin{abstract}
{ 
Studying the complexity of states sampled from various ensembles is a central component of quantum information theory. In this work we establish the average-case hardness of learning, in the statistical query model, the Born distributions of states sampled uniformly from  the compact symmetric spaces of type AI, AII, and DIII. These ensembles correspond in various senses to the so called circular and (fermionic) Gaussian ensembles.  This finding complements analogous recent results for states sampled from the classical compact groups. On the technical side, we employ a somewhat unconventional approach to integrating over the compact groups which may be of some independent interest. For example, our approach allows us to exactly evaluate the total variation distances between the output distributions of Haar random unitary and orthogonal circuits and the constant distribution, which were previously known only approximately.
}
\end{abstract}

\maketitle

\tableofcontents

\newpage
\section{Introduction}
Random ensembles of quantum states are ubiquitous in quantum computation and information~\cite{elben2022randomized,heinrich2022randomized,huang2020predicting,fisher2023random,khatri2020random,nietner2023average,schuster2024random,ragone2023unified}, as well as physics more generally~\cite{dyson1962statistical,dyson1962threefold,zirnbauer2010symmetry,heinzner2005symmetry,killip2004matrix}. While uniformly random states have in particular found numerous uses, bespoke applications are occasionally afforded by states generated by acting on a reference state with a unitary sampled uniformly from a strict subgroup of the full unitary group~\cite{hashagen2018real,west2025real,zhao2021fermionic,hearth2024efficient,wan2022matchgate,helsen2022matchgate,west2024random,gross2021schur,west2025no,grevink2025will,west2026ambient,shaya2025complexity,schatzki2024random}, or indeed from a subset which lacks a group structure~\cite{bertoni2024shallow}. For example, in the presence of various symmetries,  sampling from a well-chosen such  subset can lead to more efficient protocols for benchmarking~\cite{hashagen2018real} and tomography~\cite{west2025real,zhao2021fermionic,wan2022matchgate,helsen2022matchgate,hearth2024efficient}. \\

Applications aside, the study of the properties of states sampled from various distributions under various assumptions is naturally of considerable intrinsic interest. Suppose for example that one has access to a state $\ket\psi$ in  some fixed basis $\{\ket{\x}\}_{\x}$. As such access allows one to sample from the Born probability distribution $P_{\psi}(\x)=\abs{\braket{\x}{\psi}}^2$, one approach is to phrase questions about states drawn from an ensemble $\mce$ in terms of the distributions $P_\psi$ for $\ket\psi\sim\mce$. For instance, for a given $\mce$,  what do these distributions look like? How difficult is it to classically sample from them? If $\ket\psi$, $\ket\varphi\sim\mce$, how  (on average over $\mce$, say) do $P_\psi$ and $P_{\varphi}$ compare? Amongst such questions, in this work we ask:
\vspace{1.85mm}

\begin{center}
    \textit{What is the complexity of learning  the Born distributions of states\\ sampled uniformly  from the classical compact symmetric spaces?}
\end{center}

\vspace{1.85mm}
As it currently stands the question is somewhat ill-posed -- what do we mean by ``complexity'', for example? Here we will focus on sample-complexity in the \textit{statistical query model}~\cite{kearns1998efficient,nietner2023average,shaya2025complexity}. Concretely, this means that we have an (unknown) state sampled from a (known) distribution,  consider an observer with the ability to query the expectation value of observables diagonal in a certain (fixed) basis, and ask how many such queries they need to make in order to learn an accurate model of the (Born distribution of the) underlying  state. 
Our question is inspired by recent works which have answered the analogous question for several different ensembles of states. First, Ref.~\cite{nietner2023average} established the difficulty of, in the average-case, learning (within the statistical query model) both states drawn from the Haar measure on the full unitary group, and states produced by random brickwork quantum
circuits of various depths. Subsequently, Ref.~\cite{nietner2023free} found (via a rather different information-theoretic argument) that learning Born distributions resulting from matchgate circuits is also hard.  More recently, Ref.~\cite{shaya2025complexity} similarly found that the state distributions induced by sampling randomly from the Haar measures on the orthogonal and (unitary) symplectic groups are also average-case difficult to learn.  \\

Our specific choice of ensembles joins the   examples previously considered in this context~\cite{nietner2023average,nietner2023free,shaya2025complexity} in being physically motivated as well as reasonably amenable to analytic characterisation. Namely, we recall that the three \textit{circular ensembles} (unitary, orthogonal and symplectic) were introduced by Dyson in 1962~\cite{dyson1962statistical,dyson1962threefold} to describe the energy levels of complex quantum systems under various symmetry constraints. As by Wigner's theorem a symmetry of a  quantum mechanical system may always be represented by a unitary (possibly) multiplied by an anti-unitary operator that squares to $\pm\id$, a short list of possibilities arises. First, and simplest of all, the case when there is no anti-unitary symmetry is assigned to the \textit{circular unitary ensemble}, which is simply given by the usual Haar measure on the unitary group, and has therefore for our purposes been dealt with already in Ref.~\cite{nietner2023average}. In the case of an anti-unitary symmetry with $T^2=\id$, one has the \textit{circular orthogonal ensemble}, a random element of which may be represented as $U^\mst U$, with $U$ itself Haar random on the unitary group. Finally, in the case of an anti-unitary symmetry with $T^2=-\id$, one identifies the relevant matrix  ensemble as the \textit{circular symplectic ensemble}, a random element from which has the form $JU^\mst J^\mst U$, with $U$ Haar random on the unitary group and $J$ the symplectic form. The later two (which we would like to emphasise are, despite their names,  \textit{not} the Haar measures on the orthogonal and (unitary) symplectic groups) may also be recognised to be the uniform measures on the so-called \textit{compact symmetric spaces} of types AI and AII, respectively (more details on the compact symmetric spaces may be found in Section~\ref{sec:css}). We will additionally consider the ensemble of states generated by sampling randomly from the (so-called) symmetric space of type DIII, which in a certain sense corresponds to the uniform measure on the manifold of fermionic Gaussian states~\cite{hackl2021bosonic}. Indeed, the covariance matrix of a random fermionic Gaussian state may be sampled as  $O^\mst JO$, where $O$ is drawn from the Haar measure on the orthogonal group; up to a factor of the symplectic form $J$, this is exactly a uniformly random matrix from DIII\footnote{The manifold of bosonic Gaussian states (the symmetric space CI~\cite{hackl2021bosonic}), on the other hand, is non-compact and so does not support a uniform probability measure; we do not consider it here }. This yields an ensemble of states whose computational basis amplitudes are distributed like a column of the covariance matrix of a random fermionic Gaussian state; we elaborate further on the distinction between this ensemble and that of the fermionic Gaussian states themselves in the Discussion section. As we shall briefly discuss, the DIII case also appears in the study of superconducting quantum dots~\cite{dahlhaus2010random,grabsch2019pfaffian}.  \\

\noindent
Having introduced the ensembles we will consider, we can state our main result:

\begin{theorem}{\normalfont (Informal).} \label{thm:maini}
The output distributions of circuits drawn from the AI, AII and DIII ensembles are average-case hard to learn.
\end{theorem}

This result is similar to the findings of Refs.~\cite{nietner2023average,shaya2025complexity}, where the difficulty of the same task was established for the ensembles under consideration in those works. While the informal statement of Theorem~\ref{thm:maini} is perhaps not particularly surprising (for example, the dimension of AI as a manifold is only lower by a constant factor than that of the full unitary group), we emphasise that, quantitatively, the difficulty of the learning task is quite extreme. For example, in all of our examples we will find that even given inverse-exponentially accurate statistical queries, learning so much as a doubly-exponentially (in the logarithm of the dimension of the Hilbert space) small fraction of the possible Born distributions can take doubly-exponentially many such queries. \\

The structure of the rest of this work is as follows: In Section~\ref{sec:prelim} we review some background knowledge necessary to understand our results, briefly covering the framework of statistical query learning and some elementary aspects of the theory of integration on compact groups and symmetric spaces. Our approach to the latter is somewhat atypical for a work in quantum information theory in that, foreshadowing our later need to integrate non-polynomial functions of the matrix elements of our ensembles, we (largely) avoid the use of \textit{Weingarten calculus}~\cite{mele2023introduction,collins2006integration,weingarten1978asymptotic,collins2022weingarten,matsumoto2013weingarten}. Instead, we present an elementary treatment based on  random variables from the beta and gamma distributions. In Section~\ref{sec:uoexact} we apply the methods of Section~\ref{sec:prelim} to calculate the exact total variation distances between the Born distributions of Haar random unitary and orthogonal states and the constant distribution, for finite system dimension $d$ (to avoid dealing with annoying special cases we assume throughout that $d\ge 4$). This improves somewhat over calculations in Refs.~\cite{nietner2023average} and~\cite{shaya2025complexity}, where those quantities were respectively approximated to an additive precision of $\mco(1/\sqrt{d})$. In Section~\ref{sec:proof} we give our main technical contribution, a proof of Theorem~\ref{thm:main} (formal), a precise formalisation of Theorem~\ref{thm:maini} (informal).
In Section~\ref{sec:shallow} we investigate what can be said in the finite-depth case (i.e., not sampling uniformly randomly from the entire symmetric spaces, but rather from certain subensembles which can be realised by shallow quantum circuits). We find that in certain natural brickwork models we can obtain hardness results already at linear depth, as a result of the approximate design-formation  of these circuits~\cite{harrow2009random,baer2026random,chang2026classical,brandao2016local,haferkamp2021improved,haah2024efficient,chen2024incompressibility}.
Finally, we conclude in Section~\ref{sec:disc} with some general discussion.

\section{Preliminaries}\label{sec:prelim}
\vspace{3mm}

\subsection{Statistical query learning}\label{sec:sql}
As forewarned, our results are phrased within the framework of \textit{statistical query learning}; we now briefly review the aspects necessary to understand the paper. We largely follow, though abridge, the treatment of Ref.~\cite{nietner2023average} (to which the interested reader is referred for further details). Statistical query learning is concerned with the learning of distributions via access to the expectation values of functions on their support (as opposed, say, to access to samples from the distribution itself). Concretely, for a distribution $P$ on a set $\mcx$ and a tolerance $\tau\geq 0$, the \textit{statistical query oracle} $\mathsf{Stat}_\tau^P$ is a function
\begin{equation}
    \mathsf{Stat}_\tau^P: (\varphi:\mcx\to[-1,1]) \mapsto v
\end{equation}
such that $\abs{\expect_{x\sim P}\varphi(x)-v}\le \tau$. That is, one gives the oracle a bounded function $\varphi$, and the oracle returns its average over the distribution, up to an additive tolerance $\tau$. Naturally, one is  interested in the number of queries to $\mathsf{Stat}_\tau^P$ necessary to identify $P$ (perhaps given, say, that $P$ belongs to some set $\mcd$ of distributions on $\mcx$). Formally, given some fixed values of $\tau,\varepsilon$, we will consider the problem of $\varepsilon$-learning $\mcd$ to be that of, given access to the $\mathsf{Stat}_\tau^P$ corresponding to some unknown $P\in\mcd$, outputting a function $Q$ that is $\varepsilon$-close to $P$ in \textit{total variation distance}, where
\begin{equation}\label{eq:tvd}
   {\rm d}_{\rm TV }(P,Q) :=  \frac{1}{2}\sum_{x\in\mcx}\left|P(x) -Q(x)\right|.
\end{equation}
We will be interested in the (deterministic\footnote{As opposed to the situation where the learning algorithm itself may employ randomness}) average-case complexity of the above problem, over some fraction of $\mcd$. That is, given a measure $\mu$ on the set $\mcd$ of distributions and a success fraction $\beta$, the deterministic average-case query complexity is the minimum number of queries $q$ that a learning algorithm $\mca$ needs to make in order to ensure
\begin{equation}
    \Pr_{P\sim \mu}[\mca\ \varepsilon\text{-learns $P$ from $q$ ($\tau$-accurate) queries}]\ge\b.
\end{equation}
We will make crucial use of the following somewhat formidable-looking lemma from Ref.~\cite{nietner2023average}:
\begin{restatable}{lem}{lem1}[Lemma 1 of Ref.~\cite{nietner2023average}]
    Suppose there is a deterministic algorithm $\mca$
that $\varepsilon$-learns a fraction $\b$ of $\mcd$ with respect to $\mu$ from $q$ many $\tau$-accurate statistical queries. Then for any $Q$,
\begin{equation}\label{eq:lem1}
    q+1\ge \frac{\b-\Pr_{P\sim \mu}[{\rm d}_{\rm TV}(P,Q)\le\varepsilon+\tau]}{\displaystyle{\max_{\vph\,:\,\mcx\to[-1,1]}}\Pr_{P\sim\mu}\left[\Big\lvert\expect_{\x\sim P}[\varphi(\x)]-\expect_{\x\sim Q}[\varphi(\x)]\Big\rvert>\tau\right]} 
\end{equation}
\end{restatable}
\noindent
Let us try to understand a little of what is going on here. 
First of all, we will only ever take $Q$ to be the constant distribution $\mcc$ on the set $\mcx=\{0,1\}^n$ of computational basis output strings of $n$-qubit systems, i.e. $\mcc(\x)=1/d\ \forall \x\in\mcx$, with $d=2^n$. There are then two quantities that we are interested in calculating. The first is the size 
\begin{equation}
    \mfu := \Pr_{P\sim \mu}[{\rm d}_{\rm TV}(P,\mcc)\le\varepsilon+\tau]= \mu(\mcb_{\varepsilon+\tau}^{\rm TVD}(\mcc))
\end{equation}
of the $(\varepsilon+\tau)$-ball (with respect to the TVD) around the constant distribution. The fact that something like this term should appear and lower the sample complexity is reasonably intuitive; indeed, any distribution that is $\varepsilon$-close to $\mcc$ can be ``learnt'' simply by outputting $\mcc$ without taking any samples. The second quantity we should like to calculate, the denominator of the right hand side of Eq.~\eqref{eq:lem1}, is  the fraction of the distributions that can be distinguished from $Q$ via  querying  their oracle on some (fixed) observable $\varphi$; again taking $Q$ to be the uniform distribution on $\{0,1\}^n$ we denote it
\begin{equation}\label{eq:f}
\mff:=\max_{\vph\,:\,\{0,1\}^n\to[-1,1]}\Pr_{P\sim\mu}\left[\left\lvert\expect_{\x\sim P}[\varphi(\x)]-\frac{1}{d}\sum_{\x\in\{0,1\}^n}\varphi(\x)\right\rvert >\tau\right].
\end{equation}
As in Refs.~\cite{nietner2023average,shaya2025complexity}, we will in our applications fairly straightforwardly bound $\mff$ by means of the Weingarten calculus combined with a concentration of measure argument. The calculation of $\mfu$ will prove somewhat more challenging, involving as we shall see the expectation of a non-polynomial function of the matrix elements of our ensemble; nonetheless, we will be able to evaluate it by utilising the techniques of Section~\ref{sec:samp} and~\ref{sec:css}. Having calculated the average value, we will again proceed by a concentration of measure argument (Section~\ref{sec:com}). This will culminate in the following formalisation of Theorem~\ref{thm:maini}:

\begin{restatable}{thm}{thmmain}\label{thm:main}
{\normalfont (Formal).} 
Let $\tau>2/(d-1),\,\varepsilon\le \xi_0-10/d-\tau$. Set $\xi=\xi_0-10/d-\varepsilon-\tau$. Any algorithm that succeeds in $\varepsilon$-learning a $\beta$ fraction
of the output distributions of quantum circuits sampled uniformly from the uniform measures on AI, AII, or DIII requires at least $q$ many $\tau$-accurate statistical queries, with
\begin{equation}\label{eq:mim}
    q+1\ge \frac{\beta-2\exp\left[-(d-2)\xi^2/96\right]}{2\exp[-(d-2)(\tau-2/(d-1))^2/384]}
\end{equation}
where (for AI, AII and DIII respectively) we take the values  $\xi_0=1/e,\,1/e,\,\sqrt{2/(\pi e)}$.
\end{restatable}

Note, as in Ref.~\cite{nietner2023average}, that taking $\tau,\xi\sim 2^{-n/4}$ and $\b\sim\exp(-2^{0.49 n})$ cause the left-hand side of Eq.~\eqref{eq:mim} to be doubly-exponentially large in $n$; that is, learning a doubly-exponentially small fraction of the output distributions takes doubly-exponentially many inverse-exponentially accurate statistical queries. Interestingly, this argument is dramatically sensitive to the precise details of how accurate the oracle is (i.e. how small $\tau$ is). For example, increasing the sensitivity merely from $\tau\sim 2^{-n/4}$ to $\tau\sim 2^{-n/2}$ weakens the bound of Eq.~\eqref{eq:mim} all the way from being doubly-exponential in $n$ to being trivial. 

\subsection{Sampling randomly from the classical compact groups}\label{sec:samp}
In this section, we briefly recall some properties of matrices sampled from the uniform (Haar) measure\footnote{The existence and uniqueness of which we shall take for granted} $\mu_G$ on the classical compact groups, $G\in\{\mbu,\, \mbo,\,\mbsp\}$. Many more details may be found in e.g. Refs.~\cite{petz2004asymptotics,meckes2019random,collins2006integration,mele2023introduction}. 
Our starting point is the well-known fact that an individual column $u$ of a Haar random matrix sampled from $\mu_\mbu,\ \mu_\mbo,\ \mu_\mbsp$ is distributed like a normalised vector of i.i.d standard normal random variables over the real, complex, and quaternionic numbers respectively:

\begin{equation}\label{eq:normdists}
    g=(g_1,\dots,g_d)\ \text{i.i.d.},\qquad
\begin{cases}
g_j\sim \mcn_\mbr(0,1)  & (\mbo(d)),\\
g_j\sim \mcn_\mbc(0,1) & (\mbu (d)),\\
g_j\sim \mcn_\mbh(0,1) & (\mbsp (d/2))
\end{cases}
\qquad u=\dfrac{g}{\|g\|_2}.
\end{equation}
We will also require the joint distribution of two rows; operationally this is given by the Gram-Schmidt procedure: one (i) draws two rows $g_1,g_2$ from the appropriate normal distribution in Eq.~\eqref{eq:normdists} (ii) normalises the first to obtain $a=g_1/\|g_1\|$ (iii) orthogonalises $g_2$ with respect to $a$, $h_2=g_2-\braket{a}{g_2}a$ (iv) sets $b=h_2/\|h_2\|$. Then $a,b$ are distributed as two rows of a Haar random distribution from the corresponding $G$~\cite{petz2004asymptotics}. Now, while the constraints of unitarity imply that the distribution of a given element of a Haar-random matrix is not independent of any of the other elements, simple expressions for the distribution of its magnitude do exist. To see them, we first recall the definitions and a few properties of some common distributions:\\

\noindent
Firstly, we recall that the \textit{gamma distribution}\footnote{Note that ${\rm Gamma}(-,-)$ denotes the so-defined gamma \textit{distribution}, and  $\Gamma(-)$  the factorial-generalising gamma \textit{function} }
\({\rm Gamma}(k,\theta)\) has density \(x^{k-1}e^{-x/\theta}/(\Gamma(k)\theta^k)\) on \((0,\infty)\). Recall that if $X_i\sim\mcn_\mbr(0,k)$ then $\sum_{i=1}^N X_i^2\sim{\rm Gamma}(N/2,2k)$; it then follows from the above discussion that that distribution appears in our context as
\begin{equation}
|g_j|^2\sim
\begin{cases}
{\rm Gamma}(1,1), & (\mbu(d))\\
{\rm Gamma}\big(\frac12,2\big), & (\mbo(d))
\end{cases}\  ,
\qquad
\sum_{k=1}^d |g_k|^2\sim
\begin{cases}
{\rm Gamma}(d,1), & (\mbu(d))\\
{\rm Gamma}\big(\frac d2,2\big), & (\mbo(d))
\end{cases}
\end{equation}

\noindent
Secondly, the \textit{beta distribution} \({\rm Beta}(a,b)\)  has density \(x^{a-1}(1-x)^{b-1}/B(a,b)\), where \(B(a,b)=\Gamma(a)\Gamma(b)/\Gamma(a+b)\). Importantly for us, if $X\sim \Ga(\a,\vartheta)$ and $Y\sim \Ga(\b,\vartheta)$ are independent, then $X/(X+Y)\sim\Ba(\a,\b)$. As (in the above notation) $\abs{g_i}^2$ and $\sum_{j\neq i}\abs{g_j}^2$ are independent, it follows that~\cite{khatri2020random}
\begin{equation}
|U_{ij}|^2\sim
\begin{cases}
{\rm Beta}(1,d-1), & (\mbu(d))\\
 {\rm Beta}\big(\frac12,\frac{d-1}{2}\big), & (\mbo(d))
\end{cases}\ .
\end{equation}

\noindent
Let us use the above formulation to calculate the distribution of the magnitude squared of the (real) dot product of the column of a Haar random unitary with itself (a quantity we will in any event be later in need of)

\begin{restatable}{lem}{lemdot}
    Let $a$ be a column of a $d\times d$ Haar random unitary. Then $\abs{a^\mst a}^2\sim \Ba(1,\frac{d-1}{2})$
\end{restatable}
\begin{proof}
Let $g=x+iy$ with $x,y\in\mathbb{R}^d$ independent and having i.i.d. entries drawn from $\mcn_\mbr(0,1/2)$. Set $a=g/\|g\|$, so that $a$ is distributed like a column of a Haar random unitary. We introduce the auxiliary quantities $ A=\|x\|^2,\, B=\|y\|^2,\, R=A+B=\|g\|^2,\, C=x\cdot y,\,  D = A/R.$
Then $g^{\mathsf T}g=(x+iy)^{\mathsf T}(x+iy)=(A-B)+2iC,$ so that
\begin{equation}\label{eq:s}
|g^{\mathsf T}g|^2=(A-B)^2+4C^2=R^2-4(AB-C^2)
\quad\Longrightarrow\quad
S:=|a^{\mathsf T}a|^2\,=\,1-\frac{4(AB-C^2)}{R^2},
\end{equation}
which we note \textit{en passant} is $1-4\det(M)/(\tr[M]^2)$, where $M=\begin{psmallmatrix}
    A&C\\C&B
\end{psmallmatrix}$. 
Let us move to spherical coordinates, writing  $x=r_1u$, $y=r_2v$ with $r_1,r_2>0$ and $u,v\in\mathbb S^{d-1}$. Then
\begin{equation}
dx\,dy = r_1^{d-1}r_2^{d-1}\,dr_1\,dr_2\,d\omega(u)\,d\omega(v)
\end{equation}
and $x,y$ have joint density $\propto e^{-R}$ (independent of $u,v$). Now, let $z:=u\cdot v=\cos\theta$. The surface element on $\mathbb S^{d-1}$ is
$(\sin\theta)^{d-2}d\theta d\Omega$, and since $dz/d\theta=-\sin\theta$, the marginal of $z$ is
\begin{equation}\label{eq:z}
f_Z(z)\ \propto\ (1-z^2)^{\frac{d-3}{2}},\qquad -1<z<1.
\end{equation}
We now make the change of variables $(r_1,r_2)\mapsto (D,R)$; using
$r_1=\sqrt{DR}$, $r_2=\sqrt{(1-D)R}$, we find
\begin{equation}
\Bigl|\det\frac{\partial(r_1,r_2)}{\partial(D,R)}\Bigr|=\frac{1}{4\sqrt{D(1-D)}},
\end{equation}
so that the joint density becomes
\begin{equation}\label{eq:ARZ}
f (D,R,Z) \propto
e^{-R}R^{d-1} D^{\frac d2-1}(1-D)^{\frac d2-1} (1-Z^2)^{\frac{d-3}{2}}.
\end{equation}
Thus $D$, $R$, $Z$ are mutually independent; in particular $D$ has density $\propto D^{\frac d2-1}(1-D)^{\frac d2-1}$ on $(0,1)$.
Now, since $C=r_1r_2 z=\sqrt{AB}\,z$ and $AB=D(1-D)R^2$, we have from Eq.~\eqref{eq:s}
\begin{equation}
S \;=\; 1-4D(1-D)\bigl(1-z^2\bigr).
\end{equation}
We make one last change of variables,
\begin{equation}
U:=4D(1-D)\in(0,1),\qquad Y:=1-Z^2\in(0,1),
\end{equation}
with
\begin{equation}
f_U(u)\ \propto\ u^{\frac d2-1}(1-u)^{-1/2},\quad 0<u<1,
\end{equation}
\begin{equation}
f_Y(y)\ \propto\ y^{\frac{d-3}{2}}(1-y)^{-1/2},\quad 0<y<1,
\end{equation}
and $U\perp Y$.  Since $S=1-UY$, conditioning on $U=u$ gives
\begin{equation}
f_{S|U}(s\,|\,u)=\frac{1}{u}\,f_Y\Bigl(\frac{1-s}{u}\Bigr),\qquad 1-u<s<1.
\end{equation}
Integrating out $U$ (note that $u\in[1-s,1]$),
\begin{align}
f_S(s)
&=\int_{1-s}^{1} f_{S|U}(s|u)f_U(u)\,du\\
&\propto (1-s)^{\frac{d-3}{2}}\int_{1-s}^{1} u^{-\frac{d-1}{2}}
\Bigl(1-\frac{1-s}{u}\Bigr)^{-1/2} \cdot u^{\frac d2-1}(1-u)^{-1/2}\,du\\
&=(1-s)^{\frac{d-3}{2}} \int_{1-s}^{1} (u-1+s)^{-\frac12}(1-u)^{-\frac12}\,du\\
&= \pi (1-s)^{\frac{d-3}{2}}
\end{align}
Recalling that \({\rm Beta}(a,b)\)  has density \(x^{a-1}(1-x)^{b-1}/B(a,b)\) concludes the proof.
\end{proof}

\subsection{Classical compact symmetric spaces}\label{sec:css}
A (classical, compact, type~I) \emph{symmetric space} is a quotient $G/K$ where $G\in\{\mbu(d),\, \mbo(d),\,\mbsp(d)\}$ and $K \subset G$ is   the fixed–point set of an involutive automorphism $\sigma:G\to G$, i.e.\ $K = \{g\in G:\sigma(g)=g\}$. For example, the so-called AI symmetric space is given by taking $G=\mbu(d)$, $\sigma(U)=U^*$, whence $K=\mbo(d)$ and $G/K$ is (isomorphic to) the space of symmetric unitaries.  
An invariant probability measure on $G/K$, exists, is unique, and admits a simple sampling recipe~\cite{matsumoto2013weingarten}: draw $g\sim\mu_G$ and set
\begin{equation}
V \;=\; \sigma(g^{-1})\,g. \label{eq:symm-space-sampler}
\end{equation}
This measure is left–$K$–invariant\footnote{Indeed, if $g'=kg$ with $k\in K$ then $\sigma((g')^{-1})g'=\sigma((kg)^{-1})kg=\sigma(g^{-1})\sg(k^{-1})kg=\sigma(g^{-1})\sg(k)^{-1}kg=\sigma(g^{-1})k^{-1}kg=\sigma(g)^{-1}g$, where we use that $k$ is fixed by the (group homomorphism) $\sg$ }, and, when lifted to $G$, appears as a specific non‑uniform ensemble supported on $G$. Indeed, this allows one to reformulate the \emph{$n$\textsuperscript{th}‑order twirl}~\cite{mele2023introduction}
\begin{equation}
\mct^{(n)}_{G/K}(A) \;=\; \int_{V\sim \mu_{G/K}} V^{\otimes n}\,A\,(V^\dagger)^{\otimes n}
\end{equation}
over $G/K$ into a $2n$\textsuperscript{th}-order twirl over $G$~\cite{matsumoto2013weingarten}.
Unlike the group twirl $\mct^{(n)}_G$, which is $G$‑equivariant and projects onto the $n$\textsuperscript{th}-order commutant, $\mct^{(n)}_{G/K}$ is only $K$‑equivariant, and is not idempotent; it therefore is not a projector onto anything. As it turns out, there are exactly seven infinite families of such spaces~\cite{cartan1926sur}, given by:
\begin{table}[H]
\begin{center}
    \begin{tabular}{lccccccc}\label{tab:1}
        Type & $G$ & $\dim(G)$ & $K$ &   $\dim(G/K)$ & Involution  \\ \midrule
        AI & $\mbu(d)$ & $d^2$ & $\mbo(d)$ & $(d^2+d)/2$ & $*$  \\
        AII & $\mbu(2d)$ & $4d^2$ & $\mbsp(d)$ & $2d^2-d$ & $\Ad_J \circ *$ \\
        AIII & $\mbsu(p+q)$ & $(p+q)^2-1$ & $\mbs(\mbu(p)\times \mbu(q))$ &  $2pq$ &$\Ad_{I_{p,q}} $ \\
        \midrule
        BDI & $\mbso(p+q)$ & $ [(p+q)^2-p-q]/2 $ & $\mbs(\mbo(p)\times \mbo(q))\ $ &   $pq$ & $\Ad_{I_{p,q}} $  \\
        DIII & $\mbso(2d)$ & $2d^2-d$ & $\mbu(d)$ &  $d^2-d$ & $\Ad_{J} $ \\
        \midrule
        CI & $\mbsp(d)$ & $2d^2+d$ & $\mbu(d)$ &   $d^2+d$ & $\Ad_{J} $  \\
        CII & $\mbsp(p+q)$ & $2(p+q)^2+p+q$ & $\mbsp(p)\times \mbsp(q)$  & $4pq$ &$\Ad_{K_{p,q}} $  \\
    \end{tabular}
\end{center}
\end{table}    
    \noindent
    where $J$ is the canonical symplectic form, $I_{p,q}=\id_p \oplus (-\id_q)$, and $K_{p,q} = I_{p,q}^{\oplus 2}$. We emphasise that the precise forms of the involutions are mildly convention-dependent, and that different conventions exist. This table is adapted from Ref.~\cite{wierichs2025recursive}. In the following we will primarily concern ourselves with the spaces AI, AII, and DIII. \\

By the above discussion, it is evident that the integration of polynomial functions of the matrix elements of unitaries sampled from the classical compact symmetric spaces reduces to the integration of a polynomial (of doubled order) of matrix elements of unitaries sampled from the parent group; at which point one can rely on the well-known Weingarten calculus on the classical compact groups.  As an explicit example of this that we will later use, one has:
\begin{restatable}{lem}{ai1}\label{lem:ai1}
For any $A\in\endh$,
\begin{equation}
    \expect_{V\sim \mu_{\rm AI}}VAV^\dagger =  \frac{\tr[A]\id + A^\mst}{d+1}
\end{equation}
\end{restatable}
\begin{proof}
We use that $V\sim {\rm AI}$ is identically distributed to the product $\sg_{\rm AI}(U)\ad U=(U^*)\ad U =U^\mst U$, where $U\sim \mbu$, which yields (for a review of the graphical notation see e.g. Ref.~\cite{mele2023introduction})
    \begin{align}
        \expect_{V\sim {\rm AI}}VAV^\dagger &=  \expect_{U\sim \mbu}\left[U^\mst U A (U^\mst U)\ad\right] \\
        &=\expect_{U\sim \mbu} \img[1.8cm]{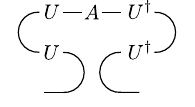}\label{eq:b5}\\
         &=  d^2\expect_{U\sim \mbu}\tr_{2,3}\left[  (\id\ot U\ts) (\ketbra{\Om}\ot A) (\id\ot (U^\dagger)\ts)(\id\ot \ketbra{\Om})\right]\\
        &=  \frac{\tr[A]\id + A^\mst}{d+1}
    \end{align}
    where $\ket{\Om}$ is a Bell state between two copies of $\mch$. Obtaining the final line is a simple exercise in the usual Weingarten calculus on the unitary group.
\end{proof}

\noindent
Similarly, we have
\begin{restatable}{lem}{aii1}\label{lem:aii1}
For any $A\in\endh$,
\begin{equation}
    \expect_{V\sim \mu_{\rm AII}}VAV^\dagger =  \frac{\tr[A]\id + J A^\mst J}{d-1}
\end{equation}
\end{restatable}
\begin{proof}
In a similar fashion to the previous lemma, we can directly calculate:
    \begin{align}
        \expect_{V\sim \mu_{\rm AII}}VAV^\dagger &=  \expect_{U\sim \mbu}\left[JU^\mst J^\mst U A (JU^\mst J^\mst U)\ad\right] \\
         &=  d^2\expect_{U\sim \mbu}\tr_{1,2}\left[ (J\ot\id\ot J) (U\ts\ot\id) (A\ot \ketbra{\Phi}) ((U^\dagger)\ts\ot\id)(J\ot\id\ot J)(\ketbra{\Phi}\ot\id)\right]\label{eq:b10}\\
        &=  \frac{\tr[A]\id +J A^\mst J}{d-1}
    \end{align}
    where Eq.~\eqref{eq:b10} follows from  a graphically-manifest manipulation similar to that of Eq.~\eqref{eq:b5}. 
\end{proof}

\noindent
More of the same:

\begin{restatable}{lem}{diii}\label{lem:diii1}
For any $A\in\eh$,
\begin{equation}
    \expect_{S\sim \mu_{\rm DIII}}SAS^\dagger =  \frac{\tr[A]\id + J A^\mst J}{d-1}
\end{equation}
\end{restatable}
\begin{proof}
We have another direct calculation:
\begin{align}
    \expect_{S\sim \mu_{\rm DIII}}SAS^\dagger &= \expect_{V\sim \mu_{\mbo}}JV^\mst J VAV^\mst J VJ\\
    &=\expect_{V\sim \mu_{\mbo}} \img[1.8cm]{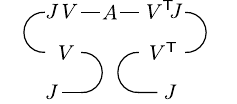}\\
    &=\sum_{abfg}J\ketbra{a}{b}J \bra{ff}(J\ot\id) \mct^{(2)}_{\mbo(d)}\left(A\ot\ketbra{a}{b}\right) (J\ot\id) \ket{gg}\\
    &=\sum_{abcefg}A_{ce}J\ketbra{a}{b} J\bra{ff}(J\ot\id) \mct^{(2)}_{\mbo(d)}\left(\ketbra{ca}{eb}\right) (J\ot\id) \ket{gg}\\
    &=\frac{1}{(d-1)(d+2)}\sum_{abce}A_{ce}J\ketbra{a}{b}J((-2-d)\delta_{ce}\delta_{ab}+(2+d)\delta_{ae}\delta_{bc})\label{eq:b17}\\
    &=\frac{1}{(d-1)}\sum_{abce}A_{ce}J\ketbra{a}{b}J(-\delta_{ce}\delta_{ab}+\delta_{ae}\delta_{bc})\\
    &=\frac{\tr[A]\id+JA^\mst J}{(d-1)}
\end{align}
Where to obtain Eq.~\eqref{eq:b17} we have employed the usual Weingarten calculus over the orthogonal group~\cite{collins2006integration,hashagen2018real,west2025real}.
\end{proof}

As it turns out, it is also possible to integrate over the symmetric spaces ``directly'', without explicitly resorting to Weingarten calculus on the parent group~\cite{matsumoto2013weingarten}. Due to our need to integrate non-polynomial functions of the matrix elements, however, we will here eschew that approach, instead mapping to a doubled-order integral on the parent group, and then using the methods of Section~\ref{sec:samp}.

\subsection{Concentration of measure}\label{sec:com}

\noindent
An important property of sampling from the classical groups is their \textit{concentration of measure}; informally, if $f:G\to\mbr$ is ``well-behaved'' then with high probability evaluating $f$ on a random $U\sim\mu_G$ will produce a value extremely close to the average value $\expval{f}:=\expect_{U\sim\mu_G}f(U)$ of $f$. More formally, say that $f$ is $L$-Lipschitz continuous on $G$ if, for all $U,V\in G$, $\abs{f(U)-f(V)}\le L\|U-V\|_2$. 
One then has~\cite{meckes2019random}: 
\begin{restatable}{lem}{lemcmg}\label{lem:cmg}
    Let $X\in\{\mbu(d),\mbo(d),\mbsp(d)\}$, and $f:X\to\mbr$ be $L$-Lipschitz continuous, with $\expect _{U\sim\mu_X} f(U) = \expval{f} $. Then
    \begin{equation}
        \Pr_{U\sim\mu_X} [\lvert f(U)  - \expval{f}\rvert\ge t] \le 2e^{-(d-2)t^2/(24L^2)} 
    \end{equation}
\end{restatable}
\begin{proof}
    This is a (slight weakening of) Theorem 5.17 of Ref.~\cite{meckes2019random}.
\end{proof}
\noindent
Our primary purpose in this subsection is to adapt this result to apply to averaging over compact symmetric spaces. We find:

\begin{restatable}{lem}{lemcms}\label{lem:cms}
    Let $G/K$ be a compact symmetric space with involution $\sg$, and $f:G\to\mbr$ be $L$-Lipschitz continuous, with $\expect_{U\sim\mu_{G/K}} f(U) = \expval{f} $. Then
    \begin{equation}
        \Pr_{U\sim\mu_{G/K}} [\lvert f(U)  - \expval{f}\rvert\ge t] \le 2e^{-(d-2)t^2/(96L^2)} 
    \end{equation}
\end{restatable}
\begin{proof}
Note that by the discussion of Section~\ref{sec:css},
\begin{equation}
        \Pr_{U\sim\mu_{G/K}} [\lvert f(U)  - \expval{f}\rvert\ge t] 
        = \Pr_{U\sim\mu_{G}} [\lvert f(\sg(U)\ad U)  - \expval{f}\rvert\ge t]; 
\end{equation}
it then follows from Lemma~\ref{lem:cmg} that we will be done if we can show that $g:U\mapsto f(\sg(U)\ad U)$ is $2L$-Lipschitz continuous. Indeed, note that by the assumed continuity of $f$,
\begin{align}
    \lvert g(U)-g(V)\rvert &=\lvert f(\sg(U)\ad U)-f(\sg(V)\ad V)\rvert\\
    &\leq L \|\sg(U)\ad U-\sg(V)\ad V\|_2\\
    &= L \|\sg(U)\ad U-\sg(U)\ad V+\sg(U)\ad V-\sg(V)\ad V\|_2\\
    &\leq L \|\sg(U)\ad U-\sg(U)\ad V\|_2 + L\| \sg(U)\ad V-\sg(V)\ad V\|_2\\
    &\leq L \| U- V\|_2 + L\| \sg(U)\ad -\sg(V)\ad \|_2\\
    &= L \| U- V\|_2 + L\| \sg(U-V)\ad \|_2\\
    &= L \| U- V\|_2 + L\| \sg(U-V) \|_2\\
    &= 2L \| U- V\|_2
\end{align}
where we have used that $\|AB\|_2\le \|A\|_2\|B\|_\infty$, and that the spectral norm of a unitary is less than or equal to one. In the final line we have used that an exhaustive examination of the cases in Table~\ref{tab:1} reveals that all of the involutions we consider are 2-norm-preserving.  
\end{proof}
One of our primary applications of Lemma~\ref{lem:cms} will be to $f:U\in G\mapsto \sbraket{0|U\ad \Phi U}{0}$, which one observes is $2\|\Phi\|_\infty$-Lipschitz continuous:
\begin{restatable}{lem}{lemexp}\label{lem:cmexp}
    Let $f:U\in G\mapsto \sbraket{0|U\ad \Phi U}{0}$, with $\Phi\in\eh$. Then $f$ is $2\|\Phi\|_\infty$-Lipschitz continuous.
\end{restatable}
\begin{proof}
We directly calculate:
\begin{align}
    \abs{f(U)-f(V)}&=\abs{\sbraket{0|U\ad \Phi U}{0}-\sbraket{0|V\ad \Phi V}{0}}\\
    &=\abs{\sbraket{0|U\ad \Phi U}{0}-\sbraket{0|V\ad \Phi U}{0}+\sbraket{0|V\ad \Phi U}{0}-\sbraket{0|V\ad \Phi V}{0}}\\
    &\le \abs{\sbraket{0|U\ad \Phi U}{0}-\sbraket{0|V\ad \Phi U}{0}}+\abs{\sbraket{0|V\ad \Phi U}{0}-\sbraket{0|V\ad \Phi V}{0}}\\
    &= \abs{\tr[\ketbra{0} (U\ad \Phi U-V\ad \Phi U)]}+\abs{\tr[\ketbra{0} ( V\ad \Phi U-V\ad \Phi V)]}\\
    &\le \|U\ad \Phi U-V\ad \Phi U\|_\infty + \|V\ad \Phi U-V\ad \Phi V\|_\infty\\
    &\le \|U\ad -V\ad \|_\infty\| \Phi U\|_\infty + \|  U-  V\|_\infty\ \|V\ad \Phi \|_\infty\\
    &\le \| \Phi \|_\infty\left( \|U\ad -V\ad \|_\infty + \|  U-  V\|_\infty\ \right)\\
    &= 2\| \Phi \|_\infty  \|U -V \|_\infty \\
    &\le  2\| \Phi \|_\infty  \|U -V \|_2
\end{align}
where we have used Holder's inequality, the submultiplicativity $\|AB\|_\infty\le\|A\|_\infty\|B\|_\infty$ of the spectral norm and that the spectral norm of a unitary is  equal to one.  
\end{proof}

 
\section{Distance to the uniform distribution in the group case}~\label{sec:uoexact}
\vspace{2mm}

In this section we apply some of the tools of Section~\ref{sec:prelim} to calculate $\expect_{U\sim\mu_G}\,{\rm d}_{\rm TV }(P_U,\mcc)$ (recall Eq.~\eqref{eq:tvd}) exactly, for $G\in\{\mbu(d),\mbo(d)\}$; previously this quantity was known only up to an additive error of order $\mco(1/\sqrt{d})$~\cite{nietner2023average,shaya2025complexity} (note that the unitary symplectic case is identical to that of the unitary case~\cite{west2024random,shaya2025complexity}, and so does not need to be handled separately).
We believe that this highlights the power of the slightly unconventional approach to integrating over the compact groups employed in this work.

\subsection{The unitary case}
Let $U\sim\mu_{\mbu}$  and $\ket{\psi} = U|0\rangle$. Equivalently, sample 
$g\sim\mcn_{\mbc}(0,\id_d)$  and set $\ket{\psi} = g/\|g\|$.
Fix $a\in[d]$ and decompose $g = \eta\,|a\rangle + g_\perp$, where $g_\perp$ and $\eta:=\braket{a}{g}\sim\mcn_{\mbc}(0,1)$ are independent. Let
\begin{equation}
E:=|\eta|^2\sim {\rm Gamma}(1,1),\qquad T_\perp:=\|g_\perp\|^2\sim {\rm Gamma}(d-1,1);    
\end{equation}
then $p_U(a) = E/(E+T_\perp)$ is (independently of $a$) distributed as $V\sim {\rm Beta}(1,d-1)$, with density
\begin{equation}
f_{V}(v) = (d-1)(1-v)^{d-2}\,\mathbf{1}_{[0,1]}(v).
\end{equation}
At this point the problem is reduced to an elementary integral:
\begin{align}
\expect_{U\sim\mu_\mbu}\,{\rm d}_{\rm TV }(P_U,\mcc)
&=\frac{1}{2}\,\expect_{V\sim {\rm Beta}(1,d-1)}\big|d V - 1\big|\\
&= \frac{1}{2}\left(\int_0^{1/d}(1-dv)f_V(v)\,{\rm d}v + \int_{1/d}^1 (dv-1)f_V(v)\,{\rm d}v\right)\\
&= \frac{d-1}{2}\left(\int_0^{1/d}(1-dv)(1-v)^{d-2}\,{\rm d}v + \int_{1/d}^1 (dv-1)(1-v)^{d-2}\,{\rm d}v\right)\\
&= \left(1-\frac{1}{d}\right)^{d}
\end{align}
whence we immediately recover the known~\cite{nietner2023average} result
\begin{equation}
    \lim_{d\to\infty} \expect_{U\sim\mu_{\mbu(d)}}\,{\rm d}_{\rm TV }(P_U,\mcc) = e^{-1}
\end{equation}

\subsection{The orthogonal case}
This time let $O\sim\mu_{\mbo(d)}$ and $\ket{\psi}=O|0\rangle$. Equivalently, sample 
$g\sim\mcn_{\mathbb{R}}(0,\id_d)$ and set $\ket{\psi}=g/\|g\|$.
Fix $a\in[d]$ and decompose $g=\eta\,|a\rangle+g_\perp$, where $g_\perp$ and $\eta:=\braket{a}{g}\sim\mcn(0,1)$  
 are independent. Consider the independent random variables
 \begin{equation}
E:=\eta^2\sim \chi^2_1,\qquad
T_\perp:=\|g_\perp\|^2\sim \chi^2_{d-1}.
 \end{equation}
Then $p_O(a)=E/(E+T_\perp)$ is distributed as $V\sim {\rm Beta}\left(\frac{1}{2},\frac{d-1}{2}\right)$. By the invariance of the Haar measure on $\mbo$,
\begin{equation}
\expect_{O\sim\mu_\mbo}\,{\rm d}_{\rm TV }(P_O,\mcc)
=\frac{1}{2}\, \expect_{V\sim {\rm Beta}\left(\frac{1}{2},\frac{d-1}{2}\right)}\big|d V - 1\big|
= \frac{1}{2}\left(\int_0^{1/d}(1-dv)f_V(v)\,{\rm d}v + \int_{1/d}^1 (dv-1)f_V(v)\,{\rm d}v\right),
\end{equation}
where $f_V(v)=\dfrac{v^{-1/2}(1-v)^{(d-3)/2}}{B(\frac{1}{2},\frac{d-1}{2})}$.
Since $\expect_{V\sim {\rm Beta}\left(\frac{1}{2},\frac{d-1}{2}\right)}[dV-1]=0$, we have
\begin{align}
\expect_{V\sim {\rm Beta}\left(\frac{1}{2},\frac{d-1}{2}\right)}\big|d V-1\big| &= 2\int_{1/d}^1(dv-1)\dfrac{v^{-1/2}(1-v)^{(d-3)/2}}{B(\frac{1}{2},\frac{d-1}{2})}\,{\rm d}v\\
&=\frac{4}{B(\frac{1}{2},\frac{d-1}{2})} \left(1-\frac{1}{d}\right)^{\frac{d-1}{2}}\frac{1}{\sqrt{d}}
\end{align}
whence
\begin{equation}
    \expect_{O\sim\mu_\mbo}\,{\rm d}_{\rm TV }(P_O,\mcc) = \frac{2\,\Gamma(\frac{d}{2})}{\sqrt{\pi d}\,\Gamma(\frac{d-1}{2})}
\left(1-\frac{1}{d}\right)^{\frac{d-1}{2}};\label{eq:oexact}
\end{equation}
standard limits then give
\begin{equation}
    \lim_{d\to\infty} \expect_{O\sim\mu_\mbo}\,{\rm d}_{\rm TV }(P_O,\mcc) =  \sqrt{\frac{2}{\pi e}}
\end{equation}
reproducing the asymptotic result of Ref.~\cite{shaya2025complexity}. Similarly to the unitary case, we remark that previously only an approximation to additive error $\mco(d^{-1/2})$ was known; in contrast the result of Eq.~\eqref{eq:oexact} is exact.

\section{Proof of Theorem 1}\label{sec:proof}
\vspace{3mm}
Our goal in this section is to prove Theorem~\ref{thm:main}, which we will do by handling the cases of AI (Theorem~\ref{thm:ai}), AII (Theorem~\ref{thm:aii}), and DIII (Theorem~\ref{thm:diii}) sequentially over the next three subsections. The general strategy (as outlined in Section~\ref{sec:sql}) is the same in each case, and there are certainly thematic overlaps between the calculations. As we shall see, however, each of them requires varying technical considerations. 

\subsection{AI}\label{sec:ai}
\noindent
The goal of this subsection is to prove Theorem~\ref{thm:ai}:
\begin{restatable}{thm}{thmai}\label{thm:ai}
    Let $\tau>2/(d+1),\,\varepsilon\le 1/e-5/d-\tau$. Set $\xi=1/e-5/d-\varepsilon-\tau$. Any algorithm that succeeds in $\varepsilon$-learning a $\beta$ fraction
of the output distributions of quantum circuits sampled uniformly from the circular orthogonal ensemble requires at least $q$ many $\tau$-accurate statistical queries, with
\begin{equation}\label{eq:aim}
    q+1\ge \frac{\beta-2\exp\left[-(d-2)\xi^2/96\right]}{2\exp[-(d-2)(\tau-2/(d+1))^2/384]}
\end{equation}
\end{restatable}

\noindent
From the discussion of Section~\ref{sec:sql} it is clear that in proving Theorem~\ref{thm:ai} there are two primary subtasks: namely, understanding the two main expressions that appear on the right hand side of Eq.~\eqref{eq:lem1}. 
We begin with:

\begin{restatable}{lem}{lemaitvd}\label{lem:aitvd}
    \begin{equation}
    \Pr_{U\sim\mu_{\rm AI}}\left[ \abs{{\rm d}_{\rm TV }(P_U,\mcc)-\frac{1}{e}} \ge \xi + \frac{5}{d}\right] \le 2e^{-(d-2)\xi^2/96} 
\end{equation}
\end{restatable}
\begin{proof}
Consider $S=U^{\mathsf T}U$, a random unitary drawn from $\mu_{\rm AI}$. Write the $0$\textsuperscript{th} and $x$\textsuperscript{th} columns of $U$ as $a,b\in\mathbb C^d$ so that
\begin{equation}
S_{x0}=\sum_{j=0}^{d-1} U_{jx}U_{j0}=a^{\mathsf T}b.
\end{equation}
We treat the diagonal $x=0$ and off--diagonal $x\neq 0$ cases separately, starting with the diagonal,
in which instance we are concerned with $S_{00}=a^{\mathsf T}a$. Put $T:=|a^{\mathsf T}a|^2\sim {\rm Beta} \left(1,\frac{d-1}{2}\right)$,
with density $f_T(t)=\frac{d-1}{2}(1-t)^{(d-3)/2}$. We then have
\begin{align}
    \expect_{T\sim {\rm Beta}\left(1,\frac{d-1}{2}\right)}|dT-1|&=\int_0^{1/d}(1-dt)f_T(t)\,{\rm d}t+\int_{1/d}^1(dt-1)f_T(t)\,{\rm d}t\\
    &=\frac{4(d-1)}{d+1}\left(1-\frac1d\right)^{\frac{d-1}{2}}-\frac{d-1}{d+1}\\
&\overset{d\to\infty}{\to} \frac{4}{\sqrt{e}}-1
\end{align}
Now we turn to the off-diagonal terms, i.e., letting $x\neq 0$. We condition on $a$. Since (recalling Section~\ref{sec:samp}) $b$ is uniform on the $(d-1)$-dimensional complex unit sphere in the subspace $a^\perp$, the scalar
\begin{equation}
W:=a^{\mathsf T}b=\langle \overline a,\,b\rangle
\end{equation}
is the projection of $b$ onto the direction of $\overline a$ \emph{within} $a^\perp$. Let
\begin{equation}
R^2:=\|P_{a^\perp}\overline a\|^2=1-|a^{\mathsf T}a|^2=1-T.
\end{equation}
Now, as we have (essentially) seen before, for complex spherical $b$ in dimension $d-1$, the  squared projection onto any fixed unit direction has a ${\rm Beta}(1,d-2)$ law. Hence, with $X\sim{\rm Beta}(1,d-2)$ independent of $a$, and defining $Z:=(1-T)\,X$, we want 
\begin{equation}
\mbe_Z\, |dZ-1|=\mathbb E_T\left[\ \mathbb E_X\,\big|\,cX-1\,\big|\ \right]    
\end{equation}
 with $c:=d(1-T)$. We begin with the inner expectation value, over $X\sim{\rm Beta}(1,d-2)$. Its density is $(d-2)(1-x)^{d-3}$ on $[0,1]$. Splitting the integral at $x=1/c$ gives, for any $c\ge0$,
\begin{align}
\expect_{X\sim{\rm Beta}(1,d-2)}|cX-1|
&=\int_0^{1/c}(1-cx)\,(d-2)(1-x)^{d-3}\,{\rm d}x+\int_{1/c}^1(cx-1)\,(d-2)(1-x)^{d-3}\,{\rm d}x\\
&=\begin{cases}
\displaystyle 1-\frac{c}{d-1}, & 0<c\le 1,\\[8pt]
\displaystyle 1-\frac{c}{d-1}+\frac{2(c-1)}{d-1}\left(1-\frac{1}{c}\right)^{d-2}, & c>1.
\end{cases}
\end{align}
It remains to average over $T$, where we recall $T\sim {\rm Beta}\bigl(1,\frac{d-1}{2}\bigr)$, so that $Y:=1-T$ has density $g(y)=\frac{d-1}{2}\,y^{(d-3)/2}$ on $[0,1]$ and $c=dY$. Decomposing the average into $c\le 1$ (i.e.\ $Y\le 1/d$) and $c>1$ (i.e.\ $Y>1/d$) we conclude that for $ x\neq  0$
\begin{align}
\expect_{U\sim{\mbu(d)}}  \left|\,d\lvert\sbraket{x}{U^\mst U|0}\rvert^2-1\right|&=\int_0^{1/d}\left(1-\frac{dy}{d-1}\right)g(y)\,{\rm d}y
+\int_{1/d}^1\left(1-\frac{dy}{d-1}+\frac{2(dy-1)}{d-1}\left(1-\frac{1}{dy}\right)^{d-2}\right)g(y)\,{\rm d}y.\\
&=\frac{1}{d+1}+\frac{(d-1)^d}{d^{\frac{d+1}{2}}}\;{}_2F_{1}\left(d,\frac{d-1}{2},\,d+1,\,-(d-1)\right)
\end{align}
where we have introduced the peculiarly denoted \textit{hypergeometric function} ${}_2F_{1}$, 
defined by 
\begin{equation}\label{eq:hyper}
    _2F_1(a,b;c;z)= \sum_{n=0}^\infty \frac{(a)_n(b)_n}{(c)_n}\frac{z^n}{n!};\qquad (q)_n=\begin{cases}
        1;&n=0\\
        q(q+1)\cdots(q+n-1); &n>0
    \end{cases}
\end{equation}
when $\abs{z}<1$ and by analytic continuation elsewhere. 
We conclude
\begin{align}
    \expect_{U\sim\mu_{\rm AI}}\,{\rm d}_{\rm TV }(P_U,\mcc) &=\frac{1}{2d}\left(\frac{4(d-1)}{d+1}\left(1-\frac1d\right)^{\frac{d-1}{2}}  +   \frac{(d-1)^{d+1}}{d^{\frac{d+1}{2}}}\;{}_2F_{1}\left(d,\frac{d-1}{2},\,d+1,\,-(d-1)\right)\right)\\
&\in \left[\frac{1}{e}-\frac{5}{d}, \frac{1}{e}+\frac{5}{d} \right]\label{eq:hga}
\end{align}
where the detailed calculations leading to the asymptotics in Eq.~\eqref{eq:hga} may be found in Lemma~\ref{lem:hga}. Now, we would like to know the probability that a given $U\sim {\rm AI}$ deviates from this average by a given amount. First, note that the function $f:\mbu(d)\to\mbr,\ f(U)={\rm d}_{\rm TV }(P_U,\mcc)$ is 1-Lipschitz continuous (Lemmas 22 of Ref.~\cite{nietner2023average} and 1 of Ref.~\cite{shaya2025complexity}), from which we conclude from Lemma~\ref{lem:cms} that, for any $\xi>0$,
\begin{equation}
    \Pr_{U\sim\mu_{\rm AI}}\left[ \abs{{\rm d}_{\rm TV }(P_U,\mcc)-\frac{1}{e}} \ge \xi + \frac{5}{d}\right] \le 2e^{-(d-2)\xi^2/96} 
\end{equation}
\end{proof}

\noindent
Next we come to the calculation of Eq.~\eqref{eq:f}; in the AI case we have
\begin{restatable}{lem}{lemaif}\label{lem:aif}
    \begin{equation}
    \mff_{\rm AI} \le 2e^{ -(d-2)(\tau-2/(d+1))^2/384}
\end{equation}
\end{restatable}

\begin{proof}
With   $\Phi = \sum_{\boldsymbol{x}\in\{0,1\}^n}\varphi(\x)\ketbra{\x}$ and $\overline{\varphi} = (1/d) \tr\Phi$ we have:
\begin{align}
    \mff_{\rm AI}&=\max_{ \vph\,:\,\{0,1\}^n\to[-1,1] }\Pr_{V\sim\mu_{\rm AI}}\left[\left\lvert  \sum_{\boldsymbol{x}\in\{0,1\}^n} \varphi(\x) \lvert\sbraket{\x}{V | 0}\rvert^2 -\overline{\vph} \right\rvert >\tau\right]\\
    &=\max_{ \vph\,:\,\{0,1\}^n\to[-1,1] }\Pr_{V\sim\mu_{\rm AI}}\left[\left\lvert   \sbraket{0}{V\ad \Phi V| 0}  -\overline{\vph} \right\rvert >\tau\right]\label{eq:d19}
\end{align}
Now, by Lemma~\ref{lem:ai1}, 
\begin{equation}
    \expect_{V\sim \mu_{\rm AI}} V\ad \Phi V = \frac{\tr[\Phi]\id + \Phi^\mst}{d+1},
\end{equation}
so that the average value of the quantity inside the absolute value on the right hand side of Eq.~\eqref{eq:d19} is bounded as
\begin{equation}
    -\frac{2}{d+1}\le \expect_{V\sim\mu_{\rm AI}} \sbraket{0}{V\ad \Phi V| 0}  -\overline{\vph} \le \frac{2}{d+1}
\end{equation}
where we have recalled that $|\varphi(\x)|\le 1$ for all $\x$. It then follows from Lemmas~\ref{lem:cms} and~\ref{lem:cmexp} that, for $\tau>2/(d+1)$, 
\begin{equation}
    \mff_{\rm AI} \le 2e^{ -(d-2)(\tau-2/(d+1))^2/384}
\end{equation}
\end{proof}

\noindent
Combining the results of Lemmas~\ref{lem:aitvd} and~\ref{lem:aif} immediately yields Theorem~\ref{thm:ai}.

\subsection{AII}
\noindent
The goal of this subsection is to prove Theorem~\ref{thm:aii}:
\begin{restatable}{thm}{thmaii}\label{thm:aii}
    Let $\tau>2/(d-1),\,\varepsilon\le 1/e-5/d-\tau$. Set $\xi=1/e-5/d-\varepsilon-\tau$. 
    Any algorithm that succeeds in $\varepsilon$-learning a $\beta$ fraction
of the output distributions of quantum circuits sampled uniformly from the circular symplectic ensemble requires at least $q$ many $\tau$-accurate statistical queries, with
\begin{equation}\label{eq:aiim}
    q+1\ge \frac{\beta-2\exp\left[-(d-2)\xi^2/96\right]}{2\exp[-(d-2)(\tau-2/(d-1))^2/384]}
\end{equation}
\end{restatable}

As we shall see, the overall structure will be essentially identical to that of the AI case, with just some minor differences emerging in the calculations. We begin with:

\begin{restatable}{lem}{lemaiitvd}\label{lem:aiitvd}
    \begin{equation}
    \Pr_{U\sim\mu_{\rm AII}}\left[ \abs{{\rm d}_{\rm TV }(P_U,\mcc)-\frac{1}{e}} \ge \xi + \frac5d\right] \le 2e^{-(d-2)\xi^2/96}  
\end{equation}
\end{restatable}

\begin{proof}
Consider $S := JU^\mst JU$, a random matrix sampled from $\mu_{\rm AII}$. 
Write $u_j$ for the $j$\textsuperscript{th} column of $U$ and for each index $x$ let $x'$ denote its ``$J$-partner'', i.e. the unique index such that $e_x^\mst J = s_x\,e_{x'}^\mst$ with $s_x\in\{\pm 1\}$. Up to an irrelevant sign, we then have
\begin{equation}
S_{x0} =  u_{x'}^\mst J u_0.
\end{equation}
We note that in the special case $x=0'$ the skew-symmetry of $J$ forces $S_{0'0} = 0$. Now, fix $x\neq 0'$ and condition on $u_0$. The remaining columns of $U$ form a random orthonormal basis of $u_0^\perp$. In particular, $u_{x'}$ is uniform on the complex unit sphere in $u_0^\perp$, which has complex dimension $d-1$. Importantly, $\overline{Ju_0}$ is also in the subspace $u_0^\perp$, so that 
\begin{equation}
    \abs{S_{x0}}^2 = \abs{u_{x'}^\mst J u_0}^2 = \abs{\braket{u_{x'}}{\overline{J u_0}}}^2
\end{equation}
is the squared overlap between a fixed unit vector and a random unit vector in a $(d-1)$-dimensional complex vector space, which we have seen to be distributed as $\mathrm{Beta}(1,\,d-2)$; as this conditional law does not depend on $u_0$, the unconditional distribution is the same, and we conclude 
 $\abs{S_{x0}}^2  \sim \mathrm{Beta}\bigl(1,d-2)$ for all $x\neq 0'$.\\

 \noindent
 So, let $X\sim\mathrm{Beta}(1,d-2)$, so $f(x)=(d-2)(1-x)^{d-3}$ on $[0,1]$; we have
\begin{align}
\expect_{X\sim\mathrm{Beta}(1,d-2)}\bigl|dX-1\bigr|
&= (d-2)\int_0^{1/d} (1-dx)(1-x)^{d-3}\,{\rm d}x
+ (d-2)\int_{1/d}^{1} (dx-1)(1-x)^{d-3}\,{\rm d}x\\
&=-\frac{1}{d-1} + 2\left(1-\frac{1}{d}\right)^{d-2}
\end{align} 
So that 
\begin{align}
\expect_{S\sim\mu_{\rm AII}}\Bigl|\,d\,|\langle x|S|0\rangle|^2-1\,\Bigr|&=
\begin{cases}
1, & x=0' \\
\displaystyle 2\left(1-\frac{1}{d}\right)^{d-2}-\frac{1}{d-1}, & x\neq 0'
\end{cases}\\
&\overset{d\to\infty}{\to}\begin{cases}
1, & x=0' \\
2/e, & x\neq 0'
\end{cases}
\end{align}
We conclude
\begin{align}
    \expect_{U\sim\mu_{\rm AII}}\,{\rm d}_{\rm TV }(P_U,\mcc) &=\frac{1}{2d}\left(1 + (d-1)\left(2\left(1-\frac{1}{d}\right)^{d-2}-\frac{1}{d-1}\right) \right)\\
    &=  \left(1-\frac{1}{d}\right)^{d-1}  \\
&\in \left[\frac{1}{e}-\frac{5}{d}, \frac{1}{e}+\frac{5}{d}\right]\label{eq:hgaii}
\end{align}
where the (elementary) details of the final bound are given in the proof of Lemma~\ref{lem:hgaii}. By identical reasoning to that employed in the proof of Lemma~\ref{lem:aitvd}, we conclude that
\begin{equation}
    \Pr_{U\sim\mu_{\rm AII}}\left[ \abs{{\rm d}_{\rm TV }(P_U,\mcc)-\frac{1}{e}} \ge \xi + \frac{5}{d}\right] \le 2e^{-(d-2)\xi^2/96} 
\end{equation}

\end{proof}

\noindent
Next we have the calculation of Eq.~\eqref{eq:f}; in the AII case we find
\begin{restatable}{lem}{lemaiif}\label{lem:aiif}
    \begin{equation}
    \mff_{\rm AII} \le 2e^{ -(d-2)(\tau-2/(d-1))^2/384}
\end{equation}
\end{restatable}

\begin{proof}
With   $\Phi = \sum_{\boldsymbol{x}\in\{0,1\}^n}\varphi(\x)\ketbra{\x}$ and $\overline{\varphi} = (1/d) \tr\Phi$ we have:
\begin{align}
    \mff_{\rm AII}&=\max_{ \vph\,:\,\{0,1\}^n\to[-1,1] }\Pr_{V\sim\mu_{\rm AII}}\left[\left\lvert  \sum_{\boldsymbol{x}\in\{0,1\}^n} \varphi(\x) \lvert\sbraket{\x}{V | 0}\rvert^2 -\overline{\vph} \right\rvert >\tau\right]\\
    &=\max_{ \vph\,:\,\{0,1\}^n\to[-1,1] }\Pr_{V\sim\mu_{\rm AII}}\left[\left\lvert   \sbraket{0}{V\ad \Phi V| 0}  -\overline{\vph} \right\rvert >\tau\right]\label{eq:d19aii}
\end{align}
Now, by Lemma~\ref{lem:aii1}, 
\begin{equation}
    \expect_{V\sim \mu_{\rm AII}} V\ad \Phi V = \frac{\tr[\Phi]\id + J \Phi^\mst J}{d-1},
\end{equation}
so that the average value of the quantity inside the absolute value on the right hand side of Eq.~\eqref{eq:d19aii} is bounded as
\begin{equation}
    -\frac{2}{d-1}\le \expect_{V\sim\mu_{\rm AII}} \sbraket{0}{V\ad \Phi V| 0}  -\overline{\vph} \le \frac{2}{d-1}
\end{equation}
where we have recalled that $|\varphi(\x)|\le 1$ for all $\x$. It then follows from Lemmas~\ref{lem:cms} and~\ref{lem:cmexp} that, for $\tau>2/(d-1)$, 
\begin{equation}
    \mff_{\rm AII} \le 2e^{ -(d-2)(\tau-2/(d-1))^2/384}
\end{equation}
\end{proof}

\noindent
Combining the results of Lemmas~\ref{lem:aiitvd} and~\ref{lem:aiif} immediately yields Theorem~\ref{thm:aii}.

 \subsection{DIII}
\noindent
The goal of this subsection is to prove Theorem~\ref{thm:diii}:
\begin{restatable}{thm}{thmdiii}\label{thm:diii}
    Let $\tau>2/(d-1),\,\varepsilon\le \sqrt{2/(\pi e)}-10/d-\tau$. Set $\xi=\sqrt{2/(\pi e)}-10/d-\varepsilon-\tau$. 
    Any algorithm that succeeds in $\varepsilon$-learning a $\beta$ fraction
of the output distributions of quantum circuits sampled uniformly from $\mu_{\rm DIII}$ requires at least $q$ many $\tau$-accurate statistical queries, with
\begin{equation}\label{eq:diiim}
    q+1\ge \frac{\beta-2\exp\left[-(d-2)\xi^2/96\right]}{2\exp[-(d-2)(\tau-2/(d-1))^2/384]}
\end{equation}
\end{restatable}
The overall structure will be essentially identical to that of the AI and AII cases, with just some minor differences emerging in the calculations. We begin with:

\begin{restatable}{lem}{lemdiiitvd}\label{lem:diiitvd}
\begin{equation}
    \Pr_{U\sim\mu_{\rm DIII}}\left[ \abs{{\rm d}_{\rm TV}(P_U,\mcc)-\sqrt\frac{2}{\pi e}}\ge\xi+\frac{10}{d}\right]\le 2e^{-(d-2)\xi^2/96}.
\end{equation}
\end{restatable}

\begin{proof}
Consider $S=JV^{\mathsf T}JV$, a random matrix from $\mu_{\rm DIII}$. We denote the columns of $V\sim\mu_{\mbo(d)}$ by $v_0,\dots,v_{d-1}$ and recall that  for each index $x$ there is a partner $x'$ with $Je_x=\pm e_{x'}$; we then have (up to an irrelevant sign) 
\begin{equation}
S_{x0}=v_{x'}^{\mathsf T}J v_0.
\end{equation}
Note that by skew-symmetry $S_{0'0}=0$.
If $x\neq 0'$, then (conditioning on $v_0$) the column $v_{x'}$ is uniform on the real unit sphere in $v_0^\perp\cong\mathbb R^{d-1}$. Because $J$ is orthogonal and $v_0^{\mathsf T} J v_0=0$, we have $Jv_0\in v_0^\perp$ and $\|Jv_0\|=1$, so $v_{x'}^{\mathsf T} J v_0$
is the first coordinate of a uniform vector on $S^{d-2}$ in $\mathbb R^{d-1}$, whence $  \abs{S_{x0}}^2 \sim \mathrm{Beta}\Bigl(\frac12,\frac{d-2}{2}\Bigr)$ for $ (x\neq 0')$.\\

\noindent
Now, with $X\sim\mathrm{Beta}(\frac12,\frac{d-2}{2})$ and $B_z(a,b)=\int_0^z u^{a-1}(1-u)^{b-1}{\rm d}u$ the \textit{incomplete Beta function}, symbolic integration yields
\begin{align}
\expect_{X\sim\mathrm{Beta}(\frac12,\frac{d-2}{2})}\bigl|dX-1\bigr|
&=\int_0^{1/d}(1-dx)\left(\frac{x^{-1/2}(1-x)^{\frac{d-4}{2}}}{B\left(\frac12,\frac{d-2}{2}\right)}\right){\rm d}x+\int_{1/d}^1(dx-1)\left(\frac{x^{-1/2}(1-x)^{\frac{d-4}{2}}}{B\left(\frac12,\frac{d-2}{2}\right)}\right){\rm d}x\\
&=-\frac{1}{d-1}+\frac{2}{B(\frac{d-2}{2},\frac12)}\left(\frac{2}{\sqrt d}\left(1-\frac{1}{d}\right)^{\frac{d-2}{2}} + B_{1-1/d}\left(\frac{d-2}{2},\frac{3}{2}\right)\right)
\end{align}
so that
\begin{align}
    \expect_{U\sim\mu_{\rm DIII}}\,{\rm d}_{\rm TV }(P_U,\mcc) &=\frac{1}{2d}\left(1 + (d-1)\left(-\frac{1}{d-1}+\frac{2}{B(\frac{d-2}{2},\frac12)}\left(\frac{2}{\sqrt d}\left(1-\frac{1}{d}\right)^{\frac{d-2}{2}} + B_{1-1/d}\left(\frac{d-2}{2},\frac{3}{2}\right)\right)\right)\right)\\
    &=\frac{(d-1)}{dB(\frac{d-2}{2},\frac12)} \left( \frac{2}{\sqrt d}\left(1-\frac{1}{d}\right)^{\frac{d-2}{2}} + B_{1-1/d}\left(\frac{d-2}{2},\frac{3}{2}\right)\right) \label{eq:e156} \\
    &\in \left[ \sqrt{\frac{2}{\pi e}} - \frac{10}{d}, \sqrt\frac{2}{\pi e} + \frac5d  \right]
\end{align}
where the detailed justification of the final line comes in Lemma~\ref{lem:diiibound}. 
We conclude that
\begin{equation}
    \Pr_{U\sim\mu_{\rm DIII}}\left[ \abs{{\rm d}_{\rm TV}(P_U,\mcc)-\sqrt\frac{2}{\pi e}}\ge\xi+\frac{10}{d}\right]\le 2e^{-(d-2)\xi^2/96}.
\end{equation}

\end{proof}

\noindent
Next we have the calculation of Eq.~\eqref{eq:f}; in the DIII case we find
\begin{restatable}{lem}{lemdiiif}\label{lem:diiif}
    \begin{equation}
    \mff_{\rm DIII} \le 2e^{ -(d-2)(\tau-2/(d-1))^2/384}
\end{equation}
\end{restatable}

\begin{proof}
With   $\Phi = \sum_{\boldsymbol{x}\in\{0,1\}^n}\varphi(\x)\ketbra{\x}$ and $\overline{\varphi} = (1/d) \tr\Phi$ we have:
\begin{align}
    \mff_{\rm DIII}&=\max_{ \vph\,:\,\{0,1\}^n\to[-1,1] }\Pr_{V\sim\mu_{\rm DIII}}\left[\left\lvert  \sum_{\boldsymbol{x}\in\{0,1\}^n} \varphi(\x) \lvert\sbraket{\x}{V | 0}\rvert^2 -\overline{\vph} \right\rvert >\tau\right]\\
    &=\max_{ \vph\,:\,\{0,1\}^n\to[-1,1] }\Pr_{V\sim\mu_{\rm DIII}}\left[\left\lvert   \sbraket{0}{V\ad \Phi V| 0}  -\overline{\vph} \right\rvert >\tau\right]\label{eq:d50}
\end{align}
Now, by Lemma~\ref{lem:diii1}, 
\begin{equation}
    \expect_{V\sim \mu_{\rm DIII}} V\ad \Phi V = \frac{\tr[\Phi]\id + J \Phi^\mst J}{d-1},
\end{equation}
so that the average value of the quantity inside the absolute value on the right hand side of Eq.~\eqref{eq:d50} is bounded as
\begin{equation}
    -\frac{2}{d-1}\le \expect_{V\sim\mu_{\rm DIII}} \sbraket{0}{V\ad \Phi V| 0}  -\overline{\vph} \le \frac{2}{d-1}
\end{equation}
where we have recalled that $|\varphi(\x)|\le 1$ for all $\x$. It then follows from Lemmas~\ref{lem:cms} and~\ref{lem:cmexp} that, for $\tau>2/(d-1)$, 
\begin{equation}
    \mff_{\rm DIII} \le 2e^{ -(d-2)(\tau-2/(d-1))^2/384}
\end{equation}
\end{proof}

\noindent
Combining the results of Lemmas~\ref{lem:diiitvd} and~\ref{lem:diiif} immediately yields Theorem~\ref{thm:diii}; combining Theorems~\ref{thm:ai},~\ref{thm:aii} and~\ref{thm:diii} yields Theorem~\ref{thm:main}. 

\section{The finite-depth case}\label{sec:shallow}
Our previous results concern elements   drawn from the uniform measures on our considered spaces; as such elements would in general require exponential (in the number $n$ of qubits we imagine to furnish our system) circuit depth to implement, however, it is natural to wonder whether similar effects can be seen at shallower depths.   Indeed, recall that a  random  element of one of our symmetric spaces looks like $\sg(g)\ad g$, where $\sg$ is the  involution defining the symmetric space, and $g$ is drawn from the parent group $G$; when we sample Haar randomly from $G$ this leads to exponentially deep circuits. 
More generally, we could imagine sampling from a different distribution $\nu$ on $G$, and checking whether our hardness results continue to hold. Indeed, we will see in this section   that, in certain natural  architectures, much of the hardness of our results can appear already at  linear depth. To that end, we show:

\begin{restatable}{thm}{thmshallow}\label{thm:shallow}
Let $\nu$ be a relative-error $\eta$-approximate 6-design over $G$ with $\eta\leq 0.01$, and let $\mu_{G/K,\nu}$ be the distribution on $G/K\in\{\rm AI,\,AII,\,DIII \}$ given by sampling $g\sim\nu$ and then returning $\sg(g)\ad g$. Suppose further that  $\tau>0,\,\varepsilon\geq 0,$ and $ \varepsilon+\tau\leq 0.1$. Then  any deterministic statistical-query  algorithm that succeeds in $\varepsilon$-learning at least a 90\% fraction of the output distributions of quantum circuits sampled from $\mu_{G/K,\nu}$  requires at least $q$ many $\tau$-accurate statistical queries, with
\begin{equation}
    q+1\ge  \frac{c\tau^2}{d^{-1}+\eta},
\end{equation}
where $c$ is a constant.
\end{restatable}
\begin{proof}
We would like to evaluate the various terms in Eq.~\eqref{eq:lem1} under the assumptions of the theorem.
Working only with 6-designs over $G$, however, we face an immediate problem. Indeed, the evaluation of the ``TVD term'' (recall Eq.~\eqref{eq:tvd}) involves, for (an average over) distributions $P$ on length-$n$ bitstrings, considering the function 
\begin{equation}
T(P)=\frac12\sum_i\big|p_i-1/d\big|=\frac1{2d}\sum_i|dp_i-1|.
\end{equation}
But  $|x-1|$ is not a polynomial function of $x$ (and certainly not of degree   at most six) so the 6-design condition does not appear to buy us very much. In the Haar-random   case we were able to nonetheless evaluate the required expectations, but for a finite brickwork architecture this is not likely to be   analytically tractable.  To get around this we introduce the cubic polynomial
\begin{equation}
h(x)=(x-1)\left(1-\frac{(x-4)^2}{8}\right)=\frac{-x^3+9x^2-16x+8}{8},
\end{equation}
whose coefficients were chosen to ensure $h(x)\leq |x-1|$ ``tightly'' (no quadratic polynomial is ``tight enough'' for the forthcoming argument). The point is we can now make the replacement
\begin{align}
T(P)&\geq\frac1{2d}\sum_i h(dp_i)\\
&=\frac{-d^3\sum_i p_i^3+9d^2\sum_i p_i^2-16d+8d}{16d}\\
&=\frac{9d\sum_i p_i^2-d^2\sum_i p_i^3-8}{16},\label{eq:tl}
\end{align}
and the approximate 6-design condition will allow us to evaluate these polynomial moments. Indeed, the distributions of the $p_i$ are (in the Haar random case) already known from Lemmas~\ref{lem:aitvd},~\ref{lem:aiitvd} and~\ref{lem:diiitvd} to be given by the various Beta distributions described therein. The   standard formula for the moments of the Beta distribution then readily yields
\begin{equation}
     \mbe_P \Big[d\sum_i p_i^2\Big] =\begin{cases}
         2+\mco(d^{-1}) & {\rm AI}\\
         2+\mco(d^{-1}) & {\rm AII}\\
         3+\mco(d^{-1}) & {\rm DIII}
     \end{cases}
\end{equation}
and
\begin{equation}
     \mbe_P \Big[d^2\sum_i p_i^3\Big] =\begin{cases}
         6+\mco(d^{-1}) & {\rm AI}\\
         6+\mco(d^{-1}) & {\rm AII}\\
         15+\mco(d^{-1}) & {\rm DIII}
     \end{cases}.
\end{equation}
Substituting these into Eq.~\eqref{eq:tl} and using that $\eta$-approximate relative-error designs match all expectation values up to a multiplicative factor of $(1\pm\eta)$~\cite{schuster2024random} then gives (in all cases) $\mbe_P[T(P)]\geq 1/4+\mco(d^{-1}+\eta)$. 
At this point we can use Markov's inequality, our hypothesis $\varepsilon+\tau\leq 0.1$ and $0\leq T(P)\leq 1$ to see
\begin{equation}
    \mfu=\Pr[T(P)\leq \varepsilon+\tau]=\Pr[1-T(P)\geq 1- \varepsilon-\tau]\leq \frac{1-\mbe_P[T]}{1- \varepsilon-\tau} \leq \frac{3/4+\mco(d^{-1}+\eta)}{0.9} = \frac{5}{6} + \mco(d^{-1}+\eta)
\end{equation}
At this point we have the numerator of Eq.~\eqref{eq:lem1} under control, and turn to its denominator. For a given $P$ as before, we consider  $F_\varphi(P)=\sum_i\varphi(i)p_i-\frac1d\sum_i\varphi(i)$
(so that the quantity for which we shall ultimately care is given by $\mathfrak f_\nu=\sup_{|\varphi|\leq1}\Pr_{P\,|\,\nu}\bigl[|F_\varphi(P)|>\tau\bigr]=\sup_{|\varphi|\leq1}\Pr_{P\,|\,\nu}\bigl[F_\varphi(P)^2>\tau^2\bigr]$).\\

Returning temporarily to the case of $P$ being induced by   the uniform measure on $G/K$, the results of Lemmas~\ref{lem:aif},~\ref{lem:aiif} and~\ref{lem:diiif} may be stated in this notation as (with $\mu$ the corresponding ensemble, and for all valid  $\varphi$), that  
$\Pr_{P\,|\, \mu}\bigl[|F_\varphi(P)-m_\varphi|>s\bigr]\leq2\exp(- (d-2)s^2/384)$, where  $\mbe_{P\,|\, \mu}F_\varphi(P)=:m_\varphi$ satisfies $|m_\varphi|\leq 2/(d-1)$. Now, the tail-integral formula gives
\begin{equation}
\operatorname{Var}_{P\,|\, \mu}[F_\varphi(P)]=\int_0^\infty2s\,\Pr_{P\,|\, \mu}\bigl[|F_\varphi(P)-m_\varphi|>s\bigr]\,{\rm d}s\leq4\int_0^\infty s\,e^{-(d-2)s^2/384}\,{\rm d}s=\frac{768}{d-2},
\end{equation}
so that
\begin{equation}
\expect_{P\,|\, \mu} [F_\varphi(P)^2]=\Var_{P\,|\, \mu}[F_\varphi(P)]+m_\varphi^2\leq\frac{768}{d-2}+\frac4{(d-1)^2}=\mco(d^{-1}).
\end{equation}
But $F_\varphi(P)^2$ is quadratic in the $p_i$,  so that averaging over an approximate 6-design approximately reproduces its Haar average. We conclude that $\expect_{P\,|\, \nu} [F_\varphi(P)^2]=\expect_{P\,|\, \mu} [F_\varphi(P)^2]+\mco(\eta)\in\mco(d^{-1}+\eta)$;  Markov's inequality then gives 
\begin{equation}
    \mff_\nu=\sup_{|\varphi|\leq1}\Pr_{P\,|\,\nu}\bigl[F_\varphi(P)^2>\tau^2\bigr]\leq\sup_{|\varphi|\leq1}\frac{\mbe_{P\,|\,\nu}F_\varphi(P)^2}{\tau^{2}}\in\mco((d^{-1}+\eta)\tau^{-2}). 
\end{equation}
Eq.~\eqref{eq:lem1} now establishes the claim of the theorem.
\end{proof}

So, Theorem~\ref{thm:shallow} shows that we can obtain hardness results when we draw $\sg(g)\ad g$, with $g$ an approximate 6-design over $G$; the question, then, is how deep a circuit one needs in order to achieve this. Remarkably, in the case of $G$ the unitary group this can be achieved already at logarithmic depth~\cite{schuster2024random,laracuente2026approximate}; in the cases of $G$ the orthogonal or symplectic group linear depth is necessary~\cite{west2026ambient,west2025no,grevink2025will} and sufficient~\cite{baer2026random,chang2026classical}. In particular, in the natural model of  drawing $g$ from a \textit{brickwork circuit} over $G$, it is known that one obtains $\eta$-approximate relative-error $k$-designs over $G$ in depth $\mco(nk+\log(1/\eta))$. Taking $k=6$ and   $\eta\sim \exp(-n)$, which can therefore be achieved in  depth $\mco(n)$, Theorem~\ref{thm:shallow} then establishes the exponential   hardness of learning even when we restrict to elements coming naturally from a linear-depth ensemble over our symmetric spaces.

\section{Discussion}~\label{sec:disc}
\vspace{2mm}

In this work we have continued the recent trend~\cite{nietner2023average,nietner2023free,shaya2025complexity} of establishing the hardness of learning within the statistical query model Born distributions of states resulting from various natural unitary ensembles. As in those works, we have found learning our target ensembles to be intractable, further highlighting the difficulty of the general problem. Similarly to Refs.~\cite{nietner2023average,shaya2025complexity} we have primarily relied upon a concentration of measure argument stemming from Levy's Lemma, albeit combined with a slightly atypical approach to carrying out the corresponding integrations; in particular, our conclusions have turned out to be reasonably insensitive the differences between our considered ensembles.\\

We note that although the conclusions of Refs.~\cite{nietner2023average,nietner2023free,shaya2025complexity} all mirror closely those of this work, the techniques employed in Ref.~\cite{nietner2023free} to establish the difficulty of learning free-fermionic (matchgate) distributions are very different from those employed in the other three papers. Indeed, their proof proceeds by a reduction to a parity-learning problem which is known to be hard. 
To see the difference between the matchgate case and our DIII analysis, we recall that a pure $d$-mode fermionic Gaussian state may be equivalently parametrised by its covariance matrix as $V=O^\mst JO$, where $O\in\mbso(2d)$ and $J$ is a (reference) \textit{complex structure}~\cite{hackl2021bosonic}, or as a matchgate circuit acting on a reference state $\ket 0$. In either case there is a redundancy of description given by a stabiliser subgroup isomorphic to $\mbu(d)$; for example in the former case any $V'$ related to $V$ by the transformation $O\mapsto O'O$, where $O'$ is an orthogonal matrix which commutes with $J$, evidently leaves the Gaussian state unchanged. As (up to the unimportant linear action of another factor of $J$), $O^\mst JO$ was exactly the quantity we considered in our DIII analysis, we see therefore that our result is then an (average-case) statement of the difficulty of learning states whose amplitudes are distributed like a column of the covariance matrix of a random fermionic Gaussian state. Interestingly, there is also a connection between the DIII case and chaotic superconducting quantum dots~\cite{dahlhaus2010random}, where the scattering matrix $S$ of a quasiparticle scattering  through the dot can, by the various symmetries present in that setup, be modelled as being sampled uniformly from DIII. That is, for some (fixed) input mode $i$, the scattering distribution   $P_S(j)=\lvert S_{i,j}|^2$ into output modes $j$ associated to the quantum dot is, by our results, difficult to learn.  

\section{Acknowledgments}

We gratefully acknowledge useful conversations with Lucas Hackl, Mart\'{i}n Larocca, Antonio Anna Mele and Marco Cerezo. This work was supported by the Laboratory Directed Research and Development program of Los Alamos National Laboratory under project number 20230049DR.  MW acknowledges the support of the Australian government research training program scholarship and the IBM Quantum Hub at the University of Melbourne.

\bibliography{refs,quantum}

\appendix

\newpage

\section{Further minutiae}\label{sec:fdetails}
\noindent
In this appendix we supply the details of some various technical assertions made throughout the text.
\begin{restatable}{lem}{lemaibound}\label{lem:hga}
\begin{equation}
    \frac{1}{e}-\frac{5}{d} \le \frac{1}{2d}\left(\frac{4(d-1)}{d+1}\left(1-\frac1d\right)^{\frac{d-1}{2}}  +   \frac{(d-1)^{d+1}}{d^{\frac{d+1}{2}}}\;{}_2F_{1}\left(d,\frac{d-1}{2},\,d+1,\,-(d-1)\right)\right) \le \frac{1}{e}+\frac{5}{d}
\end{equation}
\end{restatable}
\begin{proof}
    Let us denote the expression to be bounded by $E(d)$.
We begin by focusing on the hypergeometric piece, for which various helpful identities exist. In particular we will use \textit{Pfaff’s transformation} and \textit{Gauss' summation theorem}, given  respectively by
\begin{equation}\label{eq:pfaff}
\F(a,b;c;z)=(1-z)^{-b}\,\F\Big(c-a,\,b;\,c;\,\frac{z}{z-1}\Big)
\end{equation}
and
\begin{equation}\label{eq:gauss}
    \F(a,b;c;1)= \frac{\Gamma(c)\Gamma(c-a-b)}{\Gamma(c-a)\Gamma(c-b)};\qquad \mathfrak{Re}(c)>\mathfrak{Re}(a+b),
\end{equation}
as well as the identity
\begin{equation}
    \frac{\rm d}{{\rm d}x}{}_2F_1(a,b;c;x)=\frac{ab}{c}{}_2F_1(a+1,b+1;c+1;x) \label{eq:diffid}.
\end{equation}
First of all, by Eq.~\eqref{eq:pfaff} we have
\begin{equation}
\F\Big(d,\frac{d-1}{2};\,d+1;\,-(d-1)\Big)
=d^{-\frac{d-1}{2}}\;\F\Big(1,\frac{d-1}{2};\,d+1;\,1-\frac1d\Big).
\end{equation}
Therefore, writing
\begin{equation}
    A_d:=\left(1-\frac1d\right)^{\frac{d-1}{2}},
\qquad
B_d:=\left(1-\frac1d\right)^{d},
\qquad
S_d(x):=\F\left(1,\frac{d-1}{2};\,d+1;\,x\right),
\end{equation}
we obtain
\begin{equation}
E(d)=\frac{d-1}{2d}\left[\frac{4}{d+1}\,A_d\;+\;B_d\,S_d(1-1/d)\right].
\end{equation}
Now, on the one hand, since $c-a-b=\frac{d+1}{2}>0$ we can use Gauss’ evaluation Eq.~\eqref{eq:gauss} at $x=1$, yielding
\begin{equation}
S_d(1)=\frac{\Gamma(d+1)\Gamma(\frac{d+1}{2})}{\Gamma(d)\Gamma(\frac{d+3}{2})}
=\frac{2d}{d+1}=2-\frac{2}{d+1}.
\end{equation}
On the other, as the series Eq.~\eqref{eq:hyper} for ${}_2F_1(1,b;c;x)$ has positive coefficients,  $S_d(x)$ is increasing on $[0,1)$; using Eq.~\eqref{eq:diffid} gives
\begin{equation}
S_d'(x)=\frac{d-1}{2(d+1)}\,{}_2F_1\Big(2,\frac{d+1}{2};\,d+2;\,x\Big),
\end{equation}
whose right–hand hypergeometric function is also increasing in $x\in(0,1)$. Thus $S_d'(x)\le S_d'(1)$ and
\begin{equation}
0\le S_d(1)-S_d(1-\frac1d)\le \frac1d\,S_d'(1).
\end{equation}
By Gauss again,
\begin{equation}
{}_2F_1\Big(2,\frac{d+1}{2};\,d+2;\,1\Big)
=\frac{\Gamma(d+2)\Gamma(\frac{d-1}{2})}{\Gamma(d)\Gamma(\frac{d+3}{2})}
=\frac{4d}{d-1},
\end{equation}
so $S_d'(1)=\frac{d-1}{2(d+1)}\cdot\frac{4d}{d-1}=\frac{2d}{d+1}$, whence
\begin{equation}
 \frac{2(d-1)}{d+1}\ \le\ S_d(1-1/d)\ \le\ \frac{2d}{d+1},
\end{equation}
and
\begin{equation}
    |S_d(1-1/d)-2|\ \le\ \frac{4}{d+1}.
\end{equation}
Next we turn to bounding $A_d$ and $B_d$, which can be done in an elementary fashion. First, note that evidently $A_d\le 1$. Next, with $u=1/d\in(0,1)$ recall that $\log(1-u)=-\sum_{k\ge1}u^k/k$, with the tail satisfying
$\sum_{k\ge3} u^k/k\le  u^3/(3(1-u))$, so that 
\begin{equation}
    \log B_d =d\log(1-u)=-1-\frac{1}{2d}-\delta_d,\qquad 0\le \delta_d\le \frac{1}{3d(d-1)}< \frac{1}{d^2};
\end{equation}
exponentiating and using $|e^{-y}-1|\le y$ for $y\ge0$, we get
\begin{equation}
 |B_d-e^{-1}|\ \le\ e^{-1}\left(\frac{1}{2d}+\frac{1}{d^2}\right). 
\end{equation}
Putting everything together, we have
\begin{align}
\abs{E(d)-e^{-1}}&=\abs{\frac{d-1}{2d}\left[\frac{4}{d+1}\,A_d\;+\;B_d\,S_d(1-1/d)\right]-e^{-1}}\\
&=\abs{\frac{d-1}{2d}\left[\frac{4}{d+1}\,A_d\;+\;B_d[S_d(1-1/d)-2+2]\right]-e^{-1}}\\
&\leq  \abs{\left(\frac{d-1}{d}\right)B_d  -e^{-1}} +  \abs{\frac{d-1}{2d}\left[\frac{4}{d+1}\,A_d\;+\;B_d[S_d(1-1/d)-2]\right]}\\
&\leq  \abs{ B_d  -e^{-1}} + \frac{e^{-1}}{d} +  \abs{\frac{d-1}{2d}\frac{4}{d+1}\,A_d\;}+\frac{d-1}{2d}\abs{B_d[S_d(1-1/d)-2]}\\
&\leq  e^{-1}\left(\frac{1}{2d}+\frac{1}{d^2}\right) + \frac{e^{-1}}{d} +  \frac{d-1}{2d}\frac{4}{d+1} +\frac{d-1}{2d} \frac{4}{d+1}\\
&\leq   \frac{5}{d} 
\end{align}
\end{proof}

\begin{restatable}{lem}{lemaiibound}\label{lem:hgaii}
\begin{equation}
    \frac{1}{e}-\frac{5}{d} \le  X(d):=\left(1-\frac1d\right)^{d-1}\le \frac{1}{e}+\frac{5}{d}
\end{equation}
\end{restatable}
\begin{proof}
    Recalling that we assume $d\ge 2$, let $t=1/d\in(0,1/2]$. From the power series of \(\log(1-t)\) we have 
\begin{equation}\label{eq:series}
\Bigl(\frac1t-1\Bigr)\log(1-t) \;=\; -1 + \sum_{m=1}^{\infty}\frac{t^{m}}{m(m+1)}=:-1 + S(t),
\end{equation}
so that $X(d)=e^{-1}\exp(S(t))$. Note that
\begin{equation}
    S(t) = \sum_{m=1}^{\infty}\frac{t^{m}}{m(m+1)} \le \sum_{m=1}^{\infty}\frac{t^{m}}{2m}=\frac12\bigl(-\log(1-t)\bigr) =\frac12\int_{0}^{t}\frac{dx}{1-x}\le \frac12\int_{0}^{t}\frac{dx}{1-t}=\frac{t}{2(1-t)};
\end{equation}
note therefore that $0\le S(t)\le 1/2 $. 
Combined with the elementary inequality (valid for \(s\in[0,2)\)):
\begin{equation}\label{eq:exp}
e^{s}-1=\sum_{k=1}^{\infty}\frac{s^{k}}{k!}
\le \sum_{k=1}^{\infty}\frac{s^{k}}{2^{\,k-1}}
=\frac{s}{1-s/2},
\end{equation}
we conclude
\begin{equation}
X(d)-\frac{1}{e}
= \frac{1}{e}\bigl(e^{S(t)}-1\bigr)
\le \frac{1}{e}\cdot \frac{S(t)}{1-S(t)/2}
\le \frac{1}{e}\cdot \frac{\dfrac{t}{2(1-t)}}{1-\dfrac{t}{4(1-t)}}
= \frac{1}{e}\cdot \frac{2t}{4-5t}.
\end{equation}
Since \(t\le \frac12\), we have \(4-5t\ge \frac32\), hence
\begin{equation}
0< X(d)-\frac{1}{e}\;<\; \frac{1}{e}\cdot \frac{4}{3}\,t
= \frac{4}{3e}\cdot \frac{1}{d};
\end{equation}
this is a slightly stronger result than the statement of the lemma.

\end{proof}

\begin{restatable}{lem}{lemdiiibound}\label{lem:diiibound}
    \begin{equation}
        \expect_{U\sim\mu_{\rm DIII}}\,{\rm d}_{\rm TV }(P_U,\mcc)\in \left[ \sqrt{\frac{2}{\pi e}} - \frac{10}{d}, \sqrt\frac{2}{\pi e} + \frac5d  \right]
    \end{equation}
\end{restatable}
\begin{proof}
    We begin by recalling the \textit{regularised incomplete Beta function}, $I_z(a,b)=B_z(a,b)/B(a,b)$, which (interpreted as a function of $z$) we note to be the cumulative distribution function of the beta distribution (satisfying therefore $0\le I_z(a,b)\le 1$). Combined with the readily verified identity $B(a,1/2)=(2a+1)B(a,3/2)$, its introduction into Eq.~\eqref{eq:e156} yields
    \begin{equation}
        \expect_{U\sim\mu_{\rm DIII}}\,{\rm d}_{\rm TV }(P_U,\mcc) = \frac{2}{\sqrt{d}B(\frac{d-2}{2},\frac12)}  \left(1-\frac{1}{d}\right)^{\frac{d}{2}} +\frac{1}{d} I_{1-1/d}\left(\frac{d-2}{2},\frac{3}{2}\right) .
    \end{equation}
    By the  above characterisation of $I_z$ as a cumulative distribution function, the second term is straightforwardly seen to be in $[0,1/d]$; we therefore focus our efforts on the first. The important step is to recall \textit{Wendel's inequality}, 
    \begin{equation}\label{eq:wendel}
        \left(\frac{x}{x+s}\right)^{1-s} \le \frac{\Gamma(x+s)}{x^s\Gamma(x)}\le 1
    \end{equation}
    for $x\in\mbr^+$ and $s\in(0,1)$. In our case this yields
    \begin{equation}
        \sqrt\frac{d-2}{2\pi(1+\frac{1}{d-2})}\le \frac{1}{B(\frac{d-2}{2},\frac12)}\le\sqrt\frac{d-2}{2\pi}\label{eq:wendelfu}
    \end{equation}
    So that
    \begin{align}
        \frac{2}{\sqrt{d}B(\frac{d-2}{2},\frac12)}  \left(1-\frac{1}{d}\right)^{\frac{d}{2}} &\le \frac{2}{\sqrt{d}B(\frac{d-2}{2},\frac12)} \sqrt{\left(1-\frac{1}{d}\right)^{d}}\le 2\sqrt{\frac{1-\frac2d{}}{2\pi}}\sqrt{\frac{1}{e}+\frac{5}{d}}\le \sqrt{\frac{2}{\pi e}} \left(1+\frac{5e}{2d}\right)\le \sqrt{\frac{2}{\pi e}}  +\frac{4}{d} 
    \end{align}
    where we have recalled Lemma~\ref{lem:hgaii}, used that $(1-\frac1d)^d\le (1-\frac1d)^{d-1} $ and $\sqrt{1+x}\le1+x/2$. We conclude that
    \begin{equation}\label{eq:diiiub}
        \expect_{U\sim\mu_{\rm DIII}}\,{\rm d}_{\rm TV }(P_U,\mcc)  \le  \sqrt{\frac{2}{\pi e}}  +\frac{4}{d}        + \frac1d = \sqrt{\frac{2}{\pi e}}  +\frac{5}{d}  
    \end{equation}
    On the other hand, Lemma~\ref{lem:hgaii} also implies that\footnote{In order to keep the terms in the square roots non-negative we assume for  this next bit that $d\ge 20$; for $d<20$ our final bound is readily seen to be vacuously true }
    \begin{equation}
        \sqrt{\left(1-\frac{1}{d}\right)^{d}}\ge \sqrt{\left(1-\frac{1}{d}\right)\left(\frac 1e-\frac 5d\right)}= \sqrt{\frac 1e-\frac 5d -\frac{1}{ed}+\frac{5}{d^2}}\ge \sqrt\frac 1e \sqrt{ 1-\frac {6e}{d} }\ge \sqrt\frac 1e \left( 1-\frac {6e}{d} \right)
    \end{equation}
    Meanwhile, from Eq.~\eqref{eq:wendelfu} we have
    \begin{equation}
        \frac{2}{\sqrt{d}B(\frac{d-2}{2},\frac12)} \ge 2\sqrt\frac{1-\frac{2}{d}}{2\pi(1+\frac{1}{d-2})}\ge \frac{2}{\sqrt{2\pi}}\left(1-\frac{2}{d}\right) \left(1-\frac{1}{2(d-2)}\right)\ge \sqrt\frac{2}{\pi} \left(1-\frac{3}{d}\right)
    \end{equation}
    whence
    \begin{equation}
         \expect_{U\sim\mu_{\rm DIII}}\,{\rm d}_{\rm TV }(P_U,\mcc)  \ge \left(\sqrt\frac 1e \left( 1-\frac {6e}{d} \right)\right) \left( \sqrt\frac{2}{\pi} \left(1-\frac{3}{d}\right)\right) \ge \sqrt{\frac{2}{\pi e}}  -\frac {10}{d} ,
    \end{equation}
    which (combined with Eq.~\eqref{eq:diiiub}) establishes the result. 
\end{proof}

\end{document}